\documentclass[journal,sps]{IEEEtran}

\usepackage{amsmath}
\usepackage{amssymb}
\usepackage{amsthm}
\usepackage{mathtools}
\usepackage{bbm}
\usepackage{algorithm}
\usepackage[noend]{algpseudocode}
\usepackage{bm}
\usepackage{multirow}
\usepackage{comment}
\usepackage{booktabs}
\usepackage[caption=false,font=footnotesize]{subfig}
\usepackage{glossaries}
\usepackage{hyperref}
\usepackage[dvipsnames]{xcolor}

\usepackage{float}
\usepackage{graphicx}

\usepackage[backend=bibtex,
 bibencoding=utf8,
 style=ieee,
 dashed=false,
 doi=false
 ]{biblatex}
\usepackage{pgfplots,tikz,tikz-3dplot}
\pgfplotsset{compat=newest,
legend image code/.code={
\draw[mark repeat=2,mark phase=2]
plot coordinates {
(0cm,0cm)
(0.15cm,0cm)        
(0.3cm,0cm)         
};%
}}
\usepgfplotslibrary{fillbetween}
\usetikzlibrary{patterns,arrows,plotmarks}
\usepgfplotslibrary{groupplots}
\pgfdeclarelayer{background}
\pgfsetlayers{background,main}
\usetikzlibrary{automata,positioning}
\usetikzlibrary{decorations}
\usetikzlibrary{shapes.arrows}
\usetikzlibrary{tikzmark}
\usetikzlibrary{calc}
\usetikzlibrary{decorations.markings}
\usetikzlibrary{shapes.geometric}
\algrenewcommand\algorithmicindent{10pt}
\usepgfplotslibrary{colorbrewer}
\usepgfplotslibrary{colormaps}

\definecolor{color1}{HTML}{0011af}
\definecolor{color2}{HTML}{8819a0}
\definecolor{color3}{HTML}{bf418d}
\definecolor{color4}{HTML}{e37076}
\definecolor{color5}{HTML}{f9a256}
\definecolor{color6}{HTML}{FF0000}
\definecolor{color7}{HTML}{009F6B}
\definecolor{color8}{HTML}{00CC99}

\errorcontextlines\maxdimen

\usepackage{physics} 
\DeclareDocumentCommand\partialderivative{ s o m g d() }
{ 
    \IfBooleanTF{#1}
    {\let\fractype\flatfrac}
    {\let\fractype\frac}
    \IfNoValueTF{#4}
    {
        \IfNoValueTF{#5}
        {\fractype{\partial \IfNoValueTF{#2}{}{^{#2}}}{\partial #3\IfNoValueTF{#2}{}{^{#2}}}}
        {\fractype{\partial \IfNoValueTF{#2}{}{^{#2}}}{\partial #3\IfNoValueTF{#2}{}{^{#2}}} \argopen(#5\argclose)}
    }
    {\fractype{\partial \IfNoValueTF{#2}{}{^{#2}} #3}{\partial #4\IfNoValueTF{#2}{}{^{#2}}}\IfValueT{#5}{(#5)}}
}

\allowdisplaybreaks

\usepackage{algorithm}
\usepackage{algorithmicx}
\usepackage{algpseudocode}
\makeatletter
\def\BState{\State\hskip-\ALG@thistlm}
\makeatother
\algnewcommand\algorithmicforeach{\textbf{for each}}
\algdef{S}[FOR]{ForEach}[1]{\algorithmicforeach\ #1\ \algorithmicdo}

\makeatletter
\algnewcommand{\LineComment}[1]{\Statex \hskip\ALG@thistlm #1}
\makeatother

\usepackage{cleveref}
\crefname{equation}{}{}
\crefname{figure}{Fig.}{Figs.}
\crefname{theorem}{Theorem}{Theorems}
\crefname{definition}{Definition}{Definitions}
\crefname{lemma}{Lemma}{Lemmas}
\crefname{proposition}{Proposition}{Propositions}
\crefname{appendix}{Appendix}{Appendices}
\crefname{table}{Table}{Tables}
\crefname{corollary}{Corollary}{Corollaries}
\crefname{section}{Section}{Sections}
\crefname{algorithm}{Algorithm}{Algorithms}

\def\R{\mathbb{R}}

\def\C{\mathbb{C}}

\def\E{\mathbb{E}}
\def\Var{\text{Var}}

\newcommand{\bs}[1]{\boldsymbol{#1}}
\DeclareMathOperator*{\argmax}{\arg\max}
\DeclareMathOperator*{\argmin}{\arg\min}

\renewcommand{\Re}{\mathfrak{R}}
\renewcommand{\Im}{\mathfrak{I}}

\newtheorem{proposition}{Proposition}
\newtheorem{corollary}{Corollary}
\newtheorem{remark}{Remark}
\newtheorem{definition}{Definition}

\newacronym[plural=RISs,firstplural=Reconfigurable Intelligent Surfaces (RISs)]{ris}{RIS}{Reconfigurable Intelligent Surface}
\newacronym[plural=SPs,firstplural=scattering points (SPs)]{sp}{SP}{scattering point}
\newacronym[plural=UEs,firstplural=user equipments (UEs)]{ue}{UE}{user equipment}
\newacronym{ofdm}{OFDM}{orthogonal frequency division multiplexing}
\newacronym{upa}{UPA}{uniform planar array}
\newacronym{csi}{CSI}{channel state information}
\newacronym{snr}{SNR}{signal-to-noise-ratio}
\newacronym{isac}{ISAC}{integrated sensing and communication}
\newacronym{los}{LOS}{line-of-sight}
\newacronym{nlos}{NLOS}{non-line-of-sight}
\newacronym{music}{MUSIC}{multiple signal classification}
\newacronym{sb}{SB}{single bounce}
\newacronym{db}{DB}{double bounce}
\newacronym{aoa}{AOA}{angle of arrival}
\newacronym{aod}{AOD}{angle of departure}
\newacronym{toa}{TOA}{time of arrival}
\newacronym{crlb}{CRLB}{Cram\'er-Rao lower bound}
\newacronym{awgn}{AWGN}{additive white Gaussian noise}
\newacronym{ml}{ML}{maximum likelihood}
\newacronym{em}{EM}{expectation-maximization}
\newacronym{peb}{PEB}{position error bound}
\newacronym{deb}{DEB}{delay error bound}
\newacronym{aeb}{AEB}{angle error bound}
\newacronym{fim}{FIM}{Fisher information matrix}
\newacronym{efim}{EFIM}{equivalent Fisher information matrix}
\newacronym{ospa}{OSPA}{optimal sub-pattern assignment}
\newacronym{gospa}{GOSPA}{generalized optimal sub-pattern assignment}
\newacronym{cdf}{CDF}{cumulative distribution function}
\newacronym{roc}{ROC}{receiver operating characteristic}
\newacronym{auc}{AUC}{area under the curve}
\newacronym{rcs}{RCS}{radar cross section}
\newacronym{omp}{OMP}{orthogonal matching pursuit}
\newacronym{slam}{SLAM}{simultaneous localization and mapping}
\newacronym{mimo}{MIMO}{multiple-input multiple-output}
\newacronym{siso}{SISO}{single-input single-output}
\newacronym{bs}{BS}{base station}

\begin{document}
\pgfplotsset{
    colormap/winter,
}

\title{RIS-Assisted High Resolution Radar Sensing}

\author{%
Martin~V.~Vejling~\IEEEmembership{Student Member,~IEEE}, Hyowon Kim~\IEEEmembership{Member,~IEEE}, Christophe A.~N.~Biscio,\\Henk Wymeersch~\IEEEmembership{Fellow,~IEEE}, Petar Popovski~\IEEEmembership{Fellow,~IEEE}
\thanks{M.~V.~Vejling (mvv@\{math,es\}.aau.dk, corresponding author) is with the Dept.~of Mathematical Sciences and the Dept.~of Electronic Systems, Aalborg University, Denmark. C.~A.~N.~Biscio (christophe@math.aau.dk) is with the Dept.~of Mathematical Sciences, Aalborg University, Denmark. P.~Popovski (petarp@es.aau.dk) is with the Dept.~of Electronic Systems, Aalborg University, Denmark. H. Wymeersch (henkw@chalmers.se) is with the Dept.~of Electrical Engineering, Chalmers University of Technology, Sweden. H.~Kim (hyowon.kim@cnu.ac.kr) is with the Dept.~of Electronics Engineering, Chungnam National University, Daejeon, South Korea. The work of M.~V.~Vejling and P.~Popovski was partly funded by the Villum Investigator Grant ``WATER'' financed by the Villum Foundation, Denmark. The work of H.~Wymeersch has been supported, in part, by the SNS JU project 6G-DISAC under the EU's Horizon Europe research and innovation program under Grant Agreement No 101139130.}}

\markboth{}%
{RIS-Assisted High Resolution Radar Sensing}

\maketitle

\begin{abstract}
This paper analyzes monostatic sensing by a user equipment (UE) for a setting in which the UE is unable to resolve multiple targets due to their interference within a single resolution bin. It is shown how sensing accuracy, in terms of both detection rate and localization accuracy, can be boosted by a reconfigurable intelligent surface (RIS), which can be advantageously used to provide signal diversity and aid in resolving the targets. Specifically, assuming prior information on the presence of a cluster of targets, a RIS beam sweep procedure is used to facilitate the high resolution sensing. We derive the Cramér-Rao lower bounds (CRLBs) for channel parameter estimation and sensing and an upper bound on the detection probability. The concept of coherence is defined and analyzed theoretically. 
Then, we propose an orthogonal matching pursuit (OMP) channel estimation algorithm combined with data association to fuse the information of the non-RIS signal and the RIS signal and perform sensing. Finally, we provide numerical results to verify the potential of RIS for improving sensor resolution, and to demonstrate that the proposed methods can realize this potential for RIS-assisted high resolution sensing.
\end{abstract}

\glsresetall

\section{Introduction}
\noindent \Gls{isac} is anticipated to be a significant technological advancement in Beyond-5G and 6G communication systems \cite{Liu2022}. It is motivated by the limited availability of spectrum caused by increasing demands on key performance indicators in sensing and communication, and the potential mutual benefits between the signals of the two systems \cite{Liu2022:Fundamental}.
\Glspl{ris} can be seen as one enabling technology for \gls{isac}, offering the possibility to advantageously control the wireless propagation environment~\cite{Chepuri2023:Integrated,Bjornson2022}.

Sensing should not degrade communication performance and therefore rely on minimal requirements on energy, spectrum usage, and delay due to pilot transmissions \cite{Liu2022}. Still, sensing needs to be accurate and implementable by low-cost devices with limited hardware specifications to make it ubiquitous \cite{Cui2021:Integrating}. It can be observed that, either the feasible sensing performance is dictated by the limited sensing resources or the sensing performance requirements dictate the necessary sensing resources \cite{Liu2022}.

%
\begin{figure}
    \centering
    \includegraphics[width=0.8\linewidth]{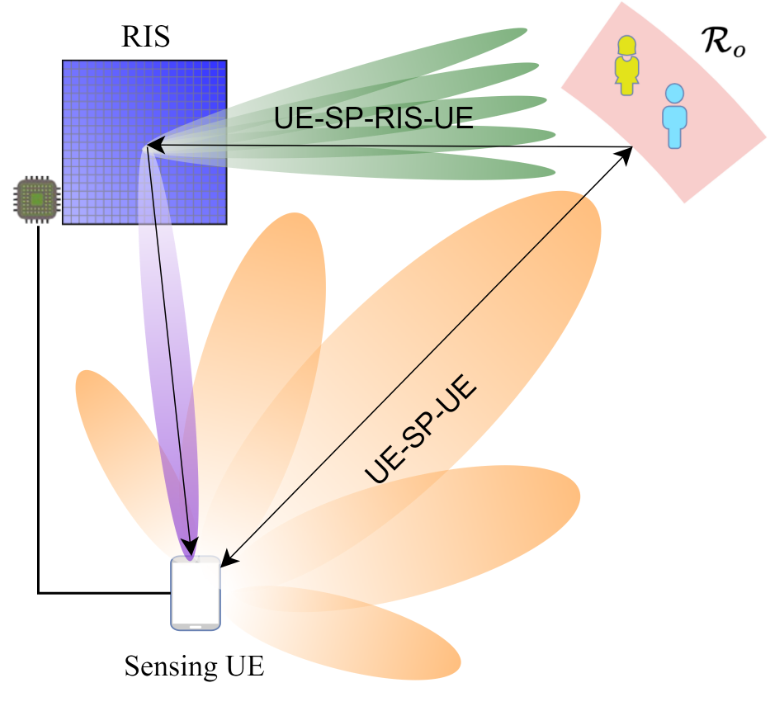}
    \caption{Illustration of the \gls{ris}-aided sensing scenario with a prior (not drawn to scale).}
    \label{fig:scenario}
\end{figure}

We consider a setup where an \gls{ris} is available and devise a full sensing protocol for the environment. The \gls{ris} can operate in two modes: (i) exclusive mode where the \gls{ue} gets associated with the \gls{ris} and can send signals via a dedicated control channel, or (ii) shared mode where the \gls{ris} changes configurations according to an announced schedule that the \gls{ue} has to use. In both cases, the monostatic sensing \gls{ue} attempts to sense the environment and decides whether to use the \gls{ris} in shared or exclusive mode.
The scenario is depicted in \cref{fig:scenario}: here a sensor with limited resolution, i.e., a low bandwidth and a small antenna array, attempts to resolve a \gls{sp} cluster which we define as a collection of \glspl{sp} within a resolution bin of the \gls{ue}. Without the \gls{ris}, only a single \gls{sp} is detected which can be critical for some sensing applications. Consider a vehicular network scenario with a human close to a car, as in \cite{Tagliaferri2023:Cooperative}: we risk causing injury if the human is not detected, however, with the support of an \gls{ris}, we can improve the resolution of the sensor \cite{Wymeersch2020}, thereby avoiding such failures. The improvement in resolution is due to two properties: (i) the \gls{ris} adds spatial diversity allowing a view of the \gls{sp} cluster from a different direction, and (ii) the \gls{ris} can have a high angular resolution \cite{Wymeersch2020}.


\subsection{Related Work}
\noindent Whilst sensing with \gls{ris} has significant potential \cite{Wymeersch2020}, there are also limitations, particularly due to the severe attenuation of double-bounce paths \cite{Tang2021:RIS}.
For this reason, most works focus on \gls{ris}-enabled scenario where: 
\gls{ris} enables a \gls{ue} to sense \glspl{sp} in \gls{nlos} through the indirect \gls{los} path via the \gls{ris} \cite{Aubry2021Surveillance,Meng2022:Sensing}; \gls{ris} provides an anchor in place of a \gls{bs} which can be used for self-localization \cite{Keykhosravi2022}; cooperative sidelink positioning is enabled by using multiple \glspl{ris} as anchor points \cite{Chen2023:Smart,Chen2024:Sidelink}; multiple \glspl{ris} with a priori unknown locations are used for \gls{slam} \cite{Yang2022MetaSLAM}; and more \cite{Zhang2022Ubiquitous}.

Still, some works explore \gls{ris}-assisted scenarios \cite{Huang2023:Harnessing,Zhang2022:MetaRadar,Hyowon2023:Bounce}. A two-step approach to \gls{slam} based on weighted least squares is presented in \cite{Huang2023:Harnessing}, exploiting the \gls{nlos} paths from active and passive scatterers to assist in localization. However, they assume perfectly known complex channel coefficients such that the \gls{ris} phase profiles can be optimized to maximize the radar cross section for sensing. This is not an assumption that can be realized in fast fading environments. A monostatic sensing scenario considering both \gls{sb} and \gls{db} multipaths is considered in \cite{Zhang2022:MetaRadar}, showing significant performance gains by exploiting the \gls{ris}, however, with the assumption that the sensing \gls{ue} is only a few wavelengths away from the \gls{ris}, which essentially requires that the \gls{ris} is attached to the \gls{ue}. A similar scenario is considered in \cite{Hyowon2023:Bounce}, however, without collocated \gls{ue} and \gls{ris}. The main takeaway is that the \gls{db} measurements provide no real performance gains.
We position our work in extension to \cite{Hyowon2023:Bounce}, and consider a scenario where using the \gls{db} signal to increase the precision of the sensing attainable with the \gls{sb} signal can provide noticeable performance gains.


The potential of \gls{ris} for sensing with low resource devices has been explored in \cite{Zhang2023:WiFi,Hu2021:MetaSensing,He2023:WiFi}: \cite{Zhang2023:WiFi} considers sensing in an indoor scenario using commodity WiFi signals and \gls{ris}, while \cite{Hu2021:MetaSensing,He2023:WiFi} investigates high-resolution capabilities of \gls{ris} in the realm of extended targets. 
The motivation of \cite{Zhang2023:WiFi} is that the low resolution properties in angle and delay of WiFi makes sensing challenging but this can be alleviated with an \gls{ris}. Considering then a dynamic indoor scenario with static interfering reflectors, they optimize the \gls{ris} phases to locate the moving targets.
%
In \cite{Hu2021:MetaSensing}, a cross-entropy loss function is optimized for the \gls{ris} phase profiles and the parameters of a neural network that maps the received signal into the estimate for the target shape. The optimization relies on a Markov decision process framework, and notably such an approach requires an extensive training procedure to calibrate the model, thereby making it incapable of adapting to non-stationary scenarios. In a related work \cite{He2023:WiFi}, the \gls{ris} focuses beams towards each point in a grid, and an optimization algorithm is proposed to estimate the target shape. This method relies on the use of an impractical number of \gls{ris} phase profiles, and the complexity of the optimization problem scales cubically with the number of grid points.
%


\subsection{Contributions and Organization}
\noindent We consider the scenario depicted in \cref{fig:scenario}, where a sensor with limited resolution attempts to resolve an \gls{sp} cluster, and devise a full sensing protocol in which the \gls{ue}, following an initial sensing step, can ask for \gls{ris} assistance in shared or exclusive mode. We study how an \gls{ris} can improve the effective sensing resolution, presenting a novel upper bound on the detection probability and the \gls{crlb} on positioning. In specific cases, we draw new insights from analytical expressions of these bounds and introduce a novel coherence concept. Using \gls{omp}, the \gls{sb} and \gls{db} parameters are estimated in parallel, followed by data association in the domain of the \gls{aoa} at the \gls{ris}, and joint position estimation, realized via weighted non-linear least squares.

Our main contributions can be summarized as follows:
\begin{enumerate}
    \item An upper bound on the detection probability of each target is derived from an oracle estimator relying on a greedy approach similar to \gls{omp}. For the specific case of three targets, analytical expressions are provided. We define a concept of coherence and show how it limits the detection probability. 
    \item A Fisher analysis of channel parameter estimation and localization is presented, deriving \glspl{deb}, \glspl{aeb}, and \glspl{peb} for both \gls{sb} and \gls{db} signals, as well as a joint \gls{peb} when combining \gls{sb}/\gls{db} information. For a simplified case with two targets, we give an analytical expression for the equivalent Fisher information of the azimuth angle parameter, discovering that the information loss of the first target due to the second target depends on the coherence.
    \item We use \gls{omp} for estimation of the \gls{sb} and \gls{db} channel parameters. Taking data association into account, we devise a joint sensing procedure for the position. Through numerical analysis, we display the potential performance gains by using an \gls{ris} for resolving an \glspl{sp} cluster. Particularly, for an \gls{sp} cluster of three targets, we are able to improve the \gls{auc} for the third target from $0.55$ without the \gls{ris} to $0.95$ with the \gls{ris}, and improve the \gls{gospa} metric by up to $60~\%$.
\end{enumerate}


The paper is organized as follows. \cref{sec:system_model} presents the scenario herein the assumptions on the prior knowledge, the design of precoder and \gls{ris} phase profiles, and the resulting signal model. The theoretical bounds on detection, channel estimation, and localization are presented in \cref{sec:theoretical_analysis}. In \cref{sec:channel_estimation}, the \gls{omp} channel estimation algorithm is recalled, followed by the proposed data association method, and the joint position estimation algorithm. Numerical results are presented in \cref{sec:numerical_experiments}, and \cref{sec:conclusion} concludes the paper and presents some future research directions.

\emph{Supplementary resources:} The code used for simulations can be found at \href{https://github.com/Martin497/RIS-Assisted-High-Resolution-Radar-Sensing}{https://github.com/Martin497/RIS-Assisted-High-Resolution-Radar-Sensing}.

\textit{Notation:} The notations $(\cdot)^\top$, $(\cdot)^H$, $\Vert \cdot \Vert$, $\langle \cdot, \cdot \rangle$, $(\cdot)^\dagger$, and ${\rm Tr}(\cdot)$ denotes the transpose, Hermitian transpose, $2$-norm, inner product, Moore-Penrose pseudo-inverse, and trace operators, respectively. We use $\otimes$ and $\odot$ for the Kronecker and Hadamard products, respectively. $\C^{N\times M}$ denotes the space of $(N \times M)$-dimensional complex matrices. The notation $\Re\{\cdot\}$ and $\Im\{\cdot\}$ refers to the real part and imaginary part of a complex number, respectively, and we use $j$ to denote the complex unit. For a matrix $\bs{A}$, the notation $[\bs{A}]_{a:b,c:d}$ refers to the submatrix of $\bs{A}$ from the $a$-th row to the $b$-th row and from the $c$-th column to the $d$-column. Moreover, for a matrix $\bs{A}$, the notation $\exp(j\bs{A})$ refers to element-wise evaluation of the complex exponential function. Finally, $\E$ and ${\rm Var}$ denotes the expectation and variance operators, respectively.

\section{System Model}\label{sec:system_model}
\subsection{Considered Scenario}
\noindent \cref{fig:scenario} depicts a scenario where a \gls{ue} attempts to sense the environment consisting of an \gls{ris} and $L$ targets also called \glspl{sp}. We assume a static scenario, so that the \gls{ue}, \gls{ris}, and \glspl{sp} have fixed positions, in a domain $\mathcal{D} \subset \R^3$. The \gls{ue}'s state is known and given as $\bs{s}=[\bs{p}^\top, \bs{o}^\top]^\top$ where $\bs{p}\in\mathcal{D}$ is the spatial position and $\bs{o}=[o_1, o_2, o_3]$ is the three-dimensional orientation vector specifying the Euler angles. Similarly, the \gls{ris} has state $\bs{s}_{\rm r} = [\bs{p}_{\rm r}^\top, \bs{o}_{\rm r}^\top]^\top$ for spatial position $\bs{p}_{\rm r} \in \mathcal{D}$ and orientation vector $\bs{o}_{\rm r}$. The \glspl{sp} constitute a point pattern $\Phi=\{\bs{c}_1, \dots, \bs{c}_L\}$ where $\bs{c}_l\in\mathcal{D}$ for $l=1,\dots,L$ and $L$ is the number of \glspl{sp} which is unknown a priori.

We assume that the \gls{ue} is full-duplex operating in \gls{ofdm} at carrier frequency $f_c$, such that the wavelength is $\lambda=c/f_c$ where $c$ is the speed of light, and subcarrier frequencies $f_n = f_c + n\Delta_f$, for $n=0,\dots,N-1$ where $N$ is the number of subcarriers and $\Delta_f$ is the subcarrier spacing. Then, the bandwidth of the sensing \gls{ue} is $W=N \Delta_f$. The \gls{ue} is equipped with a \gls{upa} consisting of $N_{\rm u} = N_{\rm u}^{\text{az}} \times N_{\rm u}^{\text{el}}$ antennas that are spaced with distance $\lambda/2$. We assume that the \gls{ris} has $N_{\rm r}=N_{\rm r}^{\text{az}} \times N_{\rm r}^{\text{el}}$ elements with distance $\lambda/4$, and that $N_{\rm r} \gg N_{\rm u}$.
To allow synchronization and control of \gls{ris}, we assume the \gls{ue} has an out-of-band channel to the \gls{ris} controller \cite{An2022:Control,Saggese2023:Control}.

\subsection{Received Signal}
\noindent The received signal for the $t$-th \gls{ofdm} symbol and $n$-th subcarrier is modeled as
\begin{equation}
    \bs{y}^{\rm sys}_{t, n} = \big(\bs{H}^{\rm los}_{n} + \bs{H}^{\rm ris}_{n}(\bs{\omega}_t)\big)\bs{f}_{t} + \bs{\varepsilon}^{\rm sys}_{t, n}
\end{equation}
for $t=1,\dots,T$ and $n=0,\dots,N-1$ where $\bs{y}^{\rm sys}_{t, n}\in\C^{N_{\rm u}}$ is the received signal,
$\bs{\omega}_t$ are the \gls{ris} phase profiles given by $\bs{\omega}_{t}=[\text{exp}(j\varphi_{t,1}),\dots,\text{exp}~(j\varphi_{t,N_{\rm r}})]^\top$ for phase-shift of the $i$-th element $\varphi_{t,i}$, $i=1,\dots,N_{\rm r}$, $\bs{H}^{\rm los}_{n} \in \C^{N_{\rm u}\times N_{\rm u}}$ is the channel matrix for the direct \gls{los} paths, $\bs{H}^{\rm ris}_{n}(\bs{\omega}_t)\in \C^{N_{\rm u}\times N_{\rm u}}$ is the channel matrix for the \gls{ris}-assisted paths, $\bs{f}_{t}\in\C^{N_{\rm u}}$ is the precoder with $\Vert \bs{f}_{t}\Vert=1$, and $\bs{\varepsilon}^{\rm sys}_{t, n}$ is circularly-symmetric complex \gls{awgn} with zero mean and covariance $\sigma^2\bs{I}$, $\sigma > 0$.

The direct \gls{los} channel matrix is modeled as
\begin{equation}
    \bs{H}^{\rm los}_{n} = \sum_{l=1}^L \underbrace{\alpha_l \bs{A}_{\rm u}(\bs{\theta}_l, \bs{\theta}_l) d_n(\tau_l)}_{\text{UE-SP}_l\text{-UE}}.
\end{equation}
In this model, $\alpha_l \in \C$ are the complex path coefficients modeling path loss due to distance and scattering as well as stochastic phase shifts and \gls{rcs} fluctuation loss due to small movements or rotations. $\bs{\theta}_l=[\theta_l^{\rm az}, \theta_l^{\rm el}]^\top$, $l=1,\dots,L$, is the azimuth and elevation \gls{aoa} at the \gls{ue} from the $l$-th \gls{sp}, $\tau_l$ are the \glspl{toa} of the paths, and $d_n(\tau) = \exp(-j2\pi n\tau\Delta_f)$ is the delay response. The array matrix is given by $\bs{A}_{\rm u}(\bs{\theta}_l, \bs{\theta}_k) = \bs{a}_{\rm u}(\bs{\theta}_l)\bs{a}_{\rm u}^\top(\bs{\theta}_k)$. The array response vector is defined by $\bs{a}_{\rm u}(\bs{\theta}) = \exp(j \bs{P}_{\rm u}^\top \bs{k}(\bs{\theta}))$ where $\bs{P}_{\rm u} = [\bs{p}_{{\rm u},1},\dots,\bs{p}_{{\rm u},N_{\rm u}}]$ with $\bs{p}_{{\rm u},i}=[x_{{\rm u},i},y_{{\rm u},i},z_{{\rm u},i}]^\top$ are the positions of the antennas in local coordinates, and $\bs{k}(\bs{\theta})=\frac{2\pi}{\lambda}[\cos(\theta^{\text{az}})\sin(\theta^{\text{el}}),\sin(\theta^{\text{az}})\sin(\theta^{\text{el}}),\cos(\theta^{\text{el}})]^{\top}$ is the wavenumber vector \cite{AbuShaban2018Bounds}.

The \gls{ris}-assisted channel matrix is modeled as \cite{Hyowon2023:Bounce}
\begin{equation}
    \begin{split}
        \bs{H}^{\rm ris}_{n}(\bs{\omega}_t) &= \underbrace{\alpha_0 \nu(\bs{\phi}_0, \bs{\phi}_0; \bs{\omega}_t)\bs{A}_{\rm u}(\bs{\theta}_0, \bs{\theta}_0)d_n(\tau_0)}_{\text{UE-RIS-UE}}\\
        &\quad+\sum_{l=1}^L \underbrace{\bar{\alpha}_l \nu(\bs{\phi}_{l}, \bs{\phi}_0; \bs{\omega}_t)\bs{A}_{\rm u}(\bs{\theta}_0, \bs{\theta}_l) d_n(\bar{\tau}_l)}_{\text{UE-SP}_l\text{-RIS-UE}}\\
        &\quad+\sum_{l=1}^L \underbrace{\bar{\alpha}_l \nu(\bs{\phi}_{l}, \bs{\phi}_0; \bs{\omega}_t)\bs{A}_{\rm u}(\bs{\theta}_l, \bs{\theta}_0)d_n(\bar{\tau}_l)}_{\text{UE-RIS-SP}_l\text{-UE}}.
    \end{split}
\end{equation}
The \gls{ris} response is $\nu(\bs{\phi}_{l}, \bs{\phi}_0; \bs{\omega}_t) = \bs{\omega}_{t}^\top \big(\bs{a}_{\rm r}(\bs{\phi}_{l}) \odot \bs{a}_{\rm r}(\bs{\phi}_0)\big)$. The parameter $\bs{\phi}_l=[\phi_l^{\rm az}, \phi_l^{\rm el}]^\top$, $l=1,\dots,L$, is the \gls{aoa} at the \gls{ris} from the $l$-th \gls{sp}, and $\bs{\phi}_0$ is the \gls{aoa} at the \gls{ris} from the \gls{ue}. Similarly, $\bs{\theta}_0$ is the \gls{aoa} at the \gls{ue} from the \gls{ris}. As before, the array response vector is defined as $\bs{a}_{\rm r}(\bs{\phi}) = \exp(j \bs{P}_{\rm r}^\top \bs{k}(\bs{\phi}))$ where $\bs{P}_{\rm r} = [\bs{p}_{{\rm r},1},\dots,\bs{p}_{{\rm r},N_{\rm r}}]$ with $\bs{p}_{{\rm r},i}=[x_{{\rm r},i},y_{{\rm r},i},z_{{\rm r},i}]^\top$ are the positions of the \gls{ris} elements in local coordinates. The complex path coefficients are $\alpha_0$ and $\bar{\alpha}_l \in \C$ and we have assumed as in \cite{Hyowon2023:Bounce} that $\bar{\alpha}_l$ is shared between both \gls{db} paths, although this assumption is not needed in our methodology. Finally, $\tau_0, \bar{\tau}_l$ models the \glspl{toa} of the paths. The specifications of all the introduced channel parameters is elaborated in Appendix \ref{subsec:channel_parameters}.

\subsection{Constructing and Using the Prior}

\noindent We assume that in an initial sensing step, a prior map of the environment is made: (i) the \gls{sp} cluster of interest has been detected and lies in resolution region $\mathcal{R}_o = \{\bs{c}\in \mathcal{D}~|~\bs{h}_{\rm n}(\bs{c}) \in \bar{\mathcal{R}}_o\}$ where $\bs{h}_{\rm n}(\bs{c})$ is the mapping from the Euclidean domain into the non-\gls{ris} channel parameter space (see Appendix \ref{subsec:channel_parameters} for details) and $\bar{\mathcal{R}}_o = \bar{\mathcal{R}}_{o,\tau} \times \bar{\mathcal{R}}_{o,\theta^{\rm az}} \times \bar{\mathcal{R}}_{o,\theta^{\rm el}}$ for resolution intervals in the delay, azimuth, and elevation angles; and (ii) the \gls{ue} and \gls{ris} positions and orientations are perfectly estimated \cite{Keykhosravi2022}.

In practice, we can construct the resolution intervals through the width of the peak in the ambiguity function \cite{Antonio2007:Ambiguity}: we form a bounding box around points that satisfy $({\rm ref} - {\rm ref}_{\rm th})$ dB, where is the value of the ambiguity function in the detection peak and ${\rm ref}_{\rm th} > 0$ is a parameter, and find the largest and smallest delay, azimuth, and elevation angles within this bounding box.
Moreover, we can estimate the \gls{ue} and \gls{ris} positions and orientations using existing techniques \cite{Keykhosravi2022}.

After the initial sensing step, the \gls{ue} requests assistance from the \gls{ris} in exclusive mode. Alternatively, if resolving the \gls{sp} cluster is not time critical or if the \gls{ris} control does not allow for real-time reservation and control, the \gls{ue} can wait for the \gls{ris} beam sweep in shared mode to cover the resolution region.

\subsection{Precoder and RIS phase profile design}
\noindent To separate the \gls{ris} and non-\gls{ris} parts of the signal, we design the \gls{ris} phase profiles as a time-orthogonal sequence, i.e., $\bs{\omega}_{2\Tilde{t}-1}=-\bs{\omega}_{2\Tilde{t}}=\Tilde{\bs{\omega}}_{\Tilde{t}}$, $\Tilde{t}=1,\dots,\Tilde{T}$, $\Tilde{T}=T/2$, for design specific \gls{ris} phase profiles $\Tilde{\bs{\omega}}_{\Tilde{t}}$. Using the prior information about the \gls{sp} cluster and the \gls{ris} position, we construct the precoder, $\bs{f}$, towards the \gls{sp} cluster and with a null towards the \gls{ris}. This setup allows us to separately observe the \gls{ue}-\gls{sp}-\gls{ue} paths and the \gls{ue}-\gls{sp}-\gls{ris}-\gls{ue} paths.\footnote{Note that the \gls{ue}-\gls{ris}-\gls{sp}-\gls{ue} paths vanish due to the precoder's null towards the RIS. We opt to use only one of the \gls{db} signals because observing both, as in \cite{Hyowon2023:Bounce}, requires dividing the time resources that halves the number of different \gls{ris} phase profiles and thereby negatively impacting sensing resolution.}


We design the \gls{ris} phase profiles $\Tilde{\bs{\omega}}_{\Tilde{t}}$ to direct beams towards sub-regions of $\mathcal{R}_{\Tilde{t}} \subset \mathcal{R}_o$, as illustrated in \cref{fig:scenario}.
Implicitly, $\mathcal{R}_o$ specifies bounds in the \gls{aoa} at the \gls{ris}, 
denoted $\phi^{\text{az}}_{\text{min}}$, $\phi^{\text{az}}_{\text{max}}$, $\phi^{\text{el}}_{\text{min}}$, and $\phi^{\text{el}}_{\text{max}}$. Given this, let $\phi^{\text{az}}_i = \phi^{\text{az}}_{\text{min}} + \zeta_{\text{az}}/2 + i \zeta_{\text{az}}$ for $i=0,\dots,D_{\text{az}}-1$ and $\phi^{\text{el}}_i = \phi^{\text{el}}_{\text{min}} + \zeta_{\text{el}}/2 + i \zeta_{\text{el}}$ for $i=0,\dots,D_{\text{el}}-1$ be sample points where $\zeta_{\text{az}} = (\phi^{\text{az}}_{\text{max}} - \phi^{\text{az}}_{\text{min}})/D_{\text{az}}$ and $\zeta_{\text{el}} = (\phi^{\text{el}}_{\text{max}} - \phi^{\text{el}}_{\text{min}})/D_{\text{el}}$ are the spacing between sample points, and $D_{\text{az}}$, $D_{\text{el}}$ satisfy $\Tilde{T}=D_{\text{az}}D_{\text{el}}$.
Then, we specify the sequence of central angles at the \gls{ris} for the directional beamforming as
\begin{equation}
    \bs{\phi}_{\Tilde{t}} = \Big[\phi^{\text{az}}_{\Tilde{t}-D_{\text{el}}(\lceil \Tilde{t}/D_{\text{el}}\rceil-1)}, \phi^{\text{el}}_{\lceil \Tilde{t}/D_{\text{az}} \rceil}\Big]^{\top}
\end{equation}
for $\Tilde{t}=1,\dots,\Tilde{T}$, and we can construct focused beams as $\Tilde{\bs{\omega}}_{\Tilde{t}} = \big(\bs{a}_{\rm r}(\bs{\phi}_{\Tilde{t}}) \odot \bs{a}_{\rm r}(\bs{\phi}_0)\big)^*$. When the \gls{ris} is large or the number of channel uses is low, these beams can fail to illuminate the entire sub-region $\mathcal{R}_{\Tilde{t}}$. To overcome this issue, we construct $\Tilde{\bs{\omega}}_{\Tilde{t}}$ as uniform beams on the region $\mathcal{R}_{\Tilde{t}}$ \cite{Kim2023:SLAM}.

\subsection{Signal Model}
\noindent Using this design of precoder, and \gls{ris} phases, we can observe the two following signals. The \gls{sb} signal, also referred to as the non-\gls{ris} signal, is expressed as
\begin{align}
    \bs{y}^{\rm n}_{t,n} &= \frac{1}{2}\big(\bs{y}^{\rm sys}_{2t-1, n} + \bs{y}^{\rm sys}_{2t, n}\big)\nonumber\\
    &= \displaystyle\sum_{l=1}^L \alpha_l d_n(\tau_l) \bs{A}_{\rm u}(\bs{\theta}_l, \bs{\theta}_l) \bs{f} + \bs{\varepsilon}_{t, n}^{\rm n}
\end{align}
for $t=1,\dots,\Tilde{T}$ where $\bs{\varepsilon}_{t, n}^{\rm n}$ is zero mean with covariance $\frac{\sigma^2}{2}\bs{I}$. The \gls{db} signal, also referred to as the \gls{ris} signal, is: 
\begin{align}
    \bs{y}^{\rm r}_{t,n} &= \frac{1}{2}\big(\bs{y}^{\rm sys}_{2t-1, n} - \bs{y}^{\rm sys}_{2t, n}\big)\nonumber\\
    &= \displaystyle\sum_{l=1}^L \bar{\alpha}_l d_n(\bar{\tau}_l) \nu(\bs{\phi}_{l}, \bs{\phi}_0; \Tilde{\bs{\omega}}_t)\bs{A}_{\rm u}(\bs{\theta}_0, \bs{\theta}_l) \bs{f} + \bs{\varepsilon}_{t, n}^{\rm r}
\end{align}
for $t=1,\dots,\Tilde{T}$; $\bs{\varepsilon}_{t, n}^{\rm r}$ has zero mean and covariance $\frac{\sigma^2}{2}\bs{I}$ \cite{Hyowon2023:Bounce}.


For modeling purposes, we vectorize the signals. We begin with the non-\gls{ris} signal by stacking the received signals first across frequency and then time $\bs{y}^{\rm n} = [(\bs{y}^{\rm n}_{1, 0})^\top, \dots, (\bs{y}^{\rm n}_{1, N-1})^\top, (\bs{y}^{\rm n}_{2, 0})^\top, \dots, (\bs{y}^{\rm n}_{\Tilde{T}, N-1})^\top]^\top$,
and we define $\bs{d}(\tau) := [d_0(\tau), \dots, d_{N-1}(\tau)]^\top$.
Then, the received signal can be expressed as
\begin{subequations}\label{eq:lin_mod}
\begin{equation}\label{eq:nonris_lin_mod}
    \bs{y}^{\rm n} = \sum_{l=1}^L \alpha_l \bs{g}^{\rm n}(\bs{\eta}_l^{\rm n}) + \bs{\varepsilon} = \bs{G}^{\rm n} \bs{\alpha} + \bs{\varepsilon}^{\rm n}
\end{equation}
where $\bs{\varepsilon}^{\rm n}$ is zero mean with covariance $\frac{\sigma^2}{2}\bs{I}$, $\bs{g}^{\rm n}_l := \bs{g}^{\rm n}(\bs{\eta}_l^{\rm n}) = \langle \bs{a}_{\rm u}^*(\bs{\theta}_l), \bs{f}\rangle \bs{1}_{\Tilde{T}} \otimes \bs{d}(\tau_l) \otimes \bs{a}_{\rm u}(\bs{\theta}_l)$, $\bs{1}_{\Tilde{T}}$ is a column vector of ones with length $\Tilde{T}$, $\bs{G}^{\rm n} := [\bs{g}^{\rm n}_1, \dots, \bs{g}^{\rm n}_L]$, $\bs{\eta}_l^{\rm n} := [\tau_l, \bs{\theta}_l^\top]^\top$, and $\bs{\alpha} := [\alpha_1,\dots,\alpha_L]^\top$.

Continuing in a similar manner with the \gls{ris} signal, we stack the received signals first across frequency and then time $\bs{y}^{\rm r} = [(\bs{y}^{\rm r}_{1, 0})^\top, \dots, (\bs{y}^{\rm r}_{1, N-1})^\top, (\bs{y}^{\rm r}_{2, 0})^\top, \dots, (\bs{y}^{\rm r}_{\Tilde{T}, N-1})^\top]^\top$,
and we define the \gls{ris} response vector $\bs{\nu}(\bs{\phi}) := [\nu(\bs{\phi}, \bs{\phi}_0; \Tilde{\bs{\omega}}_1), \dots, \nu(\bs{\phi}, \bs{\phi}_0; \Tilde{\bs{\omega}}_{\Tilde{T}})]^\top$.
Then, the received signal can be expressed as
\begin{equation}\label{eq:ris_lin_mod}
    \bs{y}^{\rm r} = \sum_{l=1}^L \bar{\alpha}_l \bs{g}^{\rm r}(\bs{\eta}_l^{\rm r}) + \bs{\varepsilon} = \bs{G}^{\rm r} \bar{\bs{\alpha}} + \bs{\varepsilon}^{\rm r}
\end{equation}
\end{subequations}
where $\bs{\varepsilon}^{\rm r}$ is again zero mean with covariance $\frac{\sigma^2}{2}\bs{I}$,
$\bs{g}^{\rm r}_l := \bs{g}^{\rm r}(\bs{\eta}_l^{\rm r}) = \langle \bs{a}_{\rm u}^*(\bs{\theta}_l), \bs{f}\rangle \bs{\nu}(\bs{\phi}_l) \otimes \bs{d}(\bar{\tau}_l) \otimes \bs{a}_{\rm u}(\bs{\theta}_0)$, $\bs{G}^{\rm r} := [\bs{g}^{\rm r}_1, \dots, \bs{g}^{\rm r}_L]$, $\bs{\eta}_l^{\rm r} := [\bs{\theta}_l^\top, \bar{\tau}_l, \bs{\phi}_l^\top]^\top$, and $\bar{\bs{\alpha}} := [\bar{\alpha}_1,\dots,\bar{\alpha}_L]^\top$.


\section{Theoretical Performance Bounds}\label{sec:theoretical_analysis}
\noindent We derive performance bounds for channel parameter estimation, position estimation, and detection. First, we consider the problem of detection, conditioning on the positions of the targets, from which we develop an upper bound on the detection probability when following a greedy approach. This result can be seen in relation to the lower bound on the probability of correct detection for the \gls{omp} algorithm introduced in \cite{Emadi2018:OMP}, although we have found this lower bound to be pessimistic relative to what is practically achieveable.
Second, conditioning on detecting the correct number of targets, we derive the \gls{crlb} on channel parameters as well as the position of the targets. 
These two theoretical results can be seen as considering individual parts of a complete sensing system.


\subsection{Detection Probability}\label{subsec:detection_probability}
\subsubsection{General Case}
In this section, we present an upper bound on the detection probability that share similarities with that of \cite{Wymeersch2020:Detection} but is generalized to simultaneous detection of multiple targets in proximity. The fundamental idea is to formulate an oracle estimator relying on a greedy approach, similar to \gls{omp}, while perfectly estimating the geometrical channel parameters.
%
\begin{proposition}\label{prop:conditional_detection_probability}
    Let the received signal be on the form $\bs{y} = \sum_{l=1}^L \alpha_l \bs{g}_l + \bs{\varepsilon} = \bs{G} \bs{\alpha} + \bs{\varepsilon}$ where $\bs{\varepsilon} \sim \mathcal{C}\mathcal{N}(\bs{0}, \frac{\sigma^2}{2}\bs{I})$, and $\bs{G} = [\bs{g}_1, \dots, \bs{g}_L]$, and $\bs{\alpha} = [\alpha_1, \dots, \alpha_L]^\top$. Then, following optimal processing with a greedy algorithm yields that the detection probability of the $l$-th target for $l > 1$, and conditioned on the channel coefficients $\bs{\alpha}_{l:L}$, is
    \begin{equation}\label{eq:conditional_detection_probability}
        p_{{\rm d}, l}(\bs{\alpha}_{l:L}) = Q_1(\sqrt{\mu_l}, \sqrt{\gamma_{\rm th}}),
    \end{equation}
    where $Q_1(\cdot, \cdot)$ is the Marcum Q-function, $\gamma_{\rm th} := -2\log(p_{\rm fa})$ for a given false alarm probability, $p_{\rm fa}$, and non-centrality parameter $\mu_l := A_l \beta_l$ for $\beta_l := |\bs{g}_l^H(\bs{I} - \bs{P}_{l-1}) \bs{G}_{l:L}\bs{\alpha}_{l:L}|^2$, $A_l := \frac{4}{\sigma^2\bs{g}_l^H(\bs{I} - \bs{P}_{l-1}) \bs{g}_l}$, $\bs{P}_{l-1} := \bs{G}_{1:l-1} \bs{G}_{1:l-1}^\dagger$.
\end{proposition}
\begin{proof}
    The proof is deferred to Appendix \ref{subsec:proof_prop1}.
\end{proof}

To summarize the detection performance across all false alarm probabilities, we will also use the \gls{auc} defined as
\begin{equation}\label{eq:conditional_AUC}
    {\rm AUC}_l(\bs{\alpha}_{l:L}) := \int_0^1 p_{{\rm d}, l}(\bs{\alpha}_{l:L}) {\rm d}p_{\rm fa},
\end{equation}
considering that the detection probability is a function of the false alarm probability. The \gls{auc} quantifies the detection performance across the range of false alarm probability, and it is lower bounded by $0.5$, corresponding to the case $p_{\rm d} = p_{\rm fa}$, and upper bounded by $1$, corresponding to the case $p_{\rm d} = 1$.

\Cref{prop:conditional_detection_probability} provides a way to compute the detection probability for both the non-\gls{ris} and \gls{ris} signals conditioned on the realization of the channel coefficients, denoted respectively by $p_{{\rm d}, l}^{\rm n}(\bs{\alpha}_{l:L})$ and $p_{{\rm d}, l}^{\rm r}(\bar{\bs{\alpha}}_{l:L})$, by considering the signal models \cref{eq:nonris_lin_mod,eq:ris_lin_mod}

By the probability of the union of not necessarily mutually exclusive events, the joint detection probability of the $l$-th target for $l > 1$, and conditioned on $\bs{\alpha}_{l:L}$, $\bar{\bs{\alpha}}_{l:L}$, is
\begin{equation}\label{eq:conditional_joint_detection_probability}
    \begin{split}
        p_{{\rm d}, l}^{\rm joint}(\bs{\alpha}_{l:L}, \bar{\bs{\alpha}}_{l:L}) &= p_{{\rm d}, l}^{\rm n}(\bs{\alpha}_{l:L}) + p_{{\rm d}, l}^{\rm r}(\bar{\bs{\alpha}}_{l:L}) \\
        &\qquad\qquad- p_{{\rm d}, l}^{\rm n}(\bs{\alpha}_{l:L}) p_{{\rm d}, l}^{\rm r}(\bar{\bs{\alpha}}_{l:L})
    \end{split}
\end{equation}
with false alarm probability $p_{\rm fa}^{\rm joint} = 2p_{\rm fa} - p_{\rm fa}^2$.

Taking the expectation over the channel coefficients with the Gaussian distribution assumption (cf.~Appendix~\ref{subsec:channel_parameters}), yields the result of the following corollary.
\begin{corollary}\label{corr:expected_detection_probability}
    With the same assumptions of \cref{prop:conditional_detection_probability}, assuming further that $\mathcal{R}_o$ is fixed, and $|\alpha_l|$ is Rayleigh distributed with scale parameter $\varsigma_{\alpha_l}$, then the marginal distribution for detection, i.e., the expected detection probability, is given by
    \begin{equation}\label{eq:expected_detection_probability}
        p_{{\rm d}, l} := \E_{\beta_l}[p_{{\rm d}, l}(\bs{\alpha}_{l:L})] = {\rm exp}\Big(\frac{{\rm log}(p_{\rm fa})}{A_l \zeta_l + 1}\Big),
    \end{equation}
    where $\zeta_l := \frac{1}{2}\bs{g}_l^H\big(\bs{I} - \bs{P}_{l-1}\big) \big(\sum_{i=l}^L \bs{g}_i\bs{g}_i^H \varsigma_{\alpha_i}^2 \big) \big(\bs{I} - \bs{P}_{l-1}\big) \bs{g}_l$. Moreover, the \gls{auc} is
    \begin{equation}\label{eq:expected_AUC}
        {\rm AUC}_l := \int_0^1 p_{{\rm d}, l} {\rm d}p_{\rm fa} = \frac{1 + A_l \zeta_l}{2 + A_l \zeta_l}.
    \end{equation}
\end{corollary}
\begin{proof}
    The proof is given in Appendix \ref{subsec:proof_corr1}.
\end{proof}
As before, this result applies straightforwardly to both the non-\gls{ris} and \gls{ris} signals, with detection probability denoted as $p_{{\rm d}, l}^{\rm n}$ and $p_{{\rm d}, l}^{\rm r}$, respectively, and \gls{auc} denoted respectively by ${\rm AUC}_{{\rm d}, l}^{\rm n}$ and ${\rm AUC}_{{\rm d}, l}^{\rm r}$.
With the previous result at hand, we can explicitly express the expected joint detection probability.
\begin{proposition}\label{prop:expected_joint_detection_probability}
    Assuming that the stochastic channel coefficients $\bs{\alpha}_{l:L}$ and $\bar{\bs{\alpha}}_{l:L}$ are independent, the expected probability of detecting the $l$-th target by either the non-\gls{ris} or \gls{ris} signals is $p_{{\rm d}, l}^{\rm joint} = p_{{\rm d}, l}^{\rm n} + p_{{\rm d}, l}^{\rm r} - p_{{\rm d}, l}^{\rm n} p_{{\rm d}, l}^{\rm r}$ with false alarm probability $p_{\rm fa}^{\rm joint} = 2p_{\rm fa} - p_{\rm fa}^2$.
\end{proposition}
\begin{proof}
    For the proof we refer to Appendix \ref{subsec:proof_prop2}.
\end{proof}
For the joint detection to be more effective than individual detections, $p_{{\rm d}, l}^{\rm joint}$ as a function of $p_{\rm fa}^{\rm joint}$, must exceed $p_{{\rm d}, l}^{\rm n}$ and $p_{{\rm d}, l}^{\rm r}$ as functions of $p_{\rm fa}$. This scenario occurs when the detection probabilities for \gls{ris} and non-\gls{ris} are similar. Conversely, if the \gls{ris} detection probability significantly surpasses the non-\gls{ris} detection probability, it is preferable to rely solely on the \gls{ris} signal, and vice versa.

In general, the projection matrices $\bs{P}_{l-1}$, and in turn the non-centrality parameters, $\mu_l$, and expected detection probabilities, $p_{{\rm d}, l}$, are not easy to quantify in terms of the signal and channel parameters, which motivates us to look into a representative case of a cluster with $3$ targets.


\subsubsection{Case study: Special case of three targets}
To give some insight into the derived detection probability, we consider the case of $L=3$, i.e., $\bs{y} = \sum_{l=1}^3 \bs{g}_l \alpha_l + \bs{\varepsilon}$ where $\bs{\varepsilon} \sim \mathcal{C}\mathcal{N}(\bs{0}, \frac{\sigma^2}{2}\bs{I})$ and $|\alpha_l|$ is Rayleigh distributed with scale parameter $\varsigma_{\alpha_l}$. This provides a simplification with sufficient generality to understand the full problem.\footnote{Considering the case of only two targets is not sufficient here: the detection probability of the second target in case of three targets depends both on the error made when subtracting the contribution of the first target, and the contribution of the third target to the signal. Such a case does not occur when considering only two targets.}

Before proceeding, we define two coherence concepts which will be central to the subsequent analysis.
\begin{definition}[Coherence]\label{def:coherence}
    The coherence $C(\bs{g}_l, \bs{g}_k) \in [0, 1]$ of two complex vectors $\bs{g}_l,\bs{g}_k\in\C^N$ is 
    \begin{equation}\label{eq:coherence}
        C(\bs{g}_l, \bs{g}_k) = \frac{|\langle \bs{g}_l, \bs{g}_k\rangle|^2}{\Vert \bs{g}_l\Vert^2 \Vert \bs{g}_k\Vert^2}.
    \end{equation}
\end{definition}
\begin{remark}
    By the Cauchy-Schwarz inequality, the coherence equals one if and only if $\bs{g}_l$ and $\bs{g}_k$ are linearly dependent, and the coherence is zero if and only if $\bs{g}_l$ and $\bs{g}_k$ are orthogonal.
\end{remark}
\begin{definition}[Generalized Coherence]\label{def:generalized_coherence}
    The generalized concept of coherence $C(\bs{g}_l, \bs{g}_k, \bs{g}_i) \in [0, 1]$ of three complex vectors $\bs{g}_l,\bs{g}_k,\bs{g}_i\in\C^N$ is
    \begin{equation}
        \Tilde{C}(\bs{g}_l, \bs{g}_k, \bs{g}_i) := \frac{\Vert \bs{g}_l \langle \bs{g}_i, \bs{g}_k\rangle - \bs{g}_k \langle \bs{g}_i, \bs{g}_l\rangle \Vert^2}{\Vert \bs{g}_l\Vert^2 \Vert \bs{g}_k\Vert^2 \Vert \bs{g}_i\Vert^2 (1 - C_{l,k})}.\label{eq:C_Tilde}
    \end{equation}
\end{definition}
For notational convenience, we define $C_{l,k} := C(\bs{g}_l, \bs{g}_k)$ as the coherence between the $l$-th and $k$-th targets, and $\Tilde{C} := \Tilde{C}(\bs{g}_1, \bs{g}_2, \bs{g}_3)$ as the generalized coherence between targets $1$, $2$, and $3$.

The non-centrality parameters, cf.~\cref{eq:conditional_detection_probability}, are
\begin{subequations}\label{eq:parameters_L3}
\begin{align}
    \mu_1 &= \frac{4 \Vert \bs{g}_1\Vert^2}{\sigma^2} \Big|\alpha_1 + \alpha_2\frac{\langle \bs{g}_1, \bs{g}_2\rangle}{\Vert \bs{g}_1\Vert^2} + \alpha_3\frac{\langle \bs{g}_1, \bs{g}_3\rangle}{\Vert \bs{g}_1\Vert^2}\Big|^2, \label{eq:parameter1_L3}\\
    \begin{split}\label{eq:parameter2_L3}
    \mu_2 &= \frac{4 \Vert \bs{g}_2 \Vert^2 (1 - C_{1,2})}{\sigma^2} \\
    &\quad \times \Big|\alpha_2 + \alpha_3 \frac{ \Vert \bs{g}_1\Vert^2\langle \bs{g}_2, \bs{g}_3\rangle - \langle \bs{g}_2, \bs{g}_1\rangle\langle \bs{g}_1, \bs{g}_3\rangle}{\Vert \bs{g}_1\Vert^2\Vert \bs{g}_2 \Vert^2(1 - C_{1,2})}\Big|^2,
    \end{split}\\
    \mu_3 &= \frac{4 \Vert \bs{g}_3\Vert^2 (1 - \Tilde{C})}{\sigma^2} |\alpha_3|^2,\label{eq:parameter3_L3}
\end{align}
\end{subequations}
where we notice that the generalized concept of coherence $\Tilde{C}$ has the same influence for the third target as the coherence $C_{1,2}$ has for the second target.
Accordingly, the expected detection probabilities are
\begin{subequations}\label{eq:expected_detection_L3}
\begin{align}
    p_{{\rm d}, 1} &= {\rm exp}\Big(\frac{\sigma^2{\rm log}(p_{\rm fa})}{4 \sum_{l=1}^3 \varsigma_{\alpha_l}^2 \Vert \bs{g}_l\Vert^2 C_{1,l} + \sigma^2}\Big), \label{eq:expected_detection_parameter1_L3}\\
    p_{{\rm d}, 2} &= {\rm exp}\Big(\frac{\sigma^2{\rm log}(p_{\rm fa})}{4 \big(\varsigma_{\alpha_2}^2 \Vert \bs{g}_2\Vert^2(1 - C_{1,2}) + \varsigma_{\alpha_3}^2 \Vert \bs{g}_3\Vert^2 \breve{C} \big) + \sigma^2}\Big), \label{eq:expected_detection_parameter2_L3}\\
    p_{{\rm d}, 3} &= {\rm exp}\Big(\frac{\sigma^2{\rm log}(p_{\rm fa})}{4 \varsigma_{\alpha_3}^2 \Vert \bs{g}_3\Vert^2 (1 - \Tilde{C}) + \sigma^2}\Big). \label{eq:expected_detection_parameter3_L3}
\end{align}
\end{subequations}
where $\breve{C} := \frac{\vert \Vert \bs{g}_1\Vert^2 \langle \bs{g}_3, \bs{g}_2\rangle - \langle \bs{g}_3, \bs{g}_1\rangle \langle \bs{g}_1, \bs{g}_2\rangle \vert^2}{\Vert \bs{g}_1\Vert^2\Vert \bs{g}_2\Vert^2\Vert \bs{g}_3\Vert^2 (1 - C_{1,2})} = \Tilde{C} - C_{1, 3}$,
and notice again the appearance of $\Tilde{C}$.
The non-centrality parameters and the expected detection probabilities applies to both the non-\gls{ris} and \gls{ris} signals using the signal models \cref{eq:nonris_lin_mod} and \cref{eq:ris_lin_mod}, respectively.

Some interpretation follows by considering some special cases. Consider first the case that $\bs{g}_2$ and $\bs{g}_3$ are linearly dependent (a special case is when targets two and three coincide), then \cref{eq:expected_detection_L3} collapses to the two target case, which we deduce by observing that $\breve{C} = 1 - C_{1,2}$ and $\Tilde{C} = 1$. Hence, the expected detection performance of the first two targets is not degraded, however, the probability of detecting the third target is equal to the false alarm probability, which is the worst case scenario for detection.
Another notable interpretation is that the non-centrality parameter for the second target, see~\cref{eq:parameter2_L3}, depends linearly on the complement of the coherence between the first and second targets, $1 - C_{1,2}$.

Considering alternatively a scenario where $\bs{g}_1, \bs{g}_2, \bs{g}_3$ are pair-wise orthogonal. Then, the expected detection probabilities, cf.~\cref{eq:expected_detection_L3}, simplify to
$p_{{\rm d}, l} = {\rm exp}\big(\frac{\sigma^2{\rm log}(p_{\rm fa})}{4 \varsigma_{\alpha_l}^2 \Vert \bs{g}_l\Vert^2 + \sigma^2}\big)$, i.e., the detection probabilities depends only on the expected receive \gls{snr} of the $l$-th path, $2\varsigma_{\alpha_l}^2 \Vert \bs{g}_l\Vert^2/\sigma^2$, and the false alarm probability.

Overall, the expected detection probabilities depend mainly on the coherence factors, with good detection performance for all targets if the coherence factors are close to zero, and degraded detection performance for the second and third targets, if the coherence factors are close to one. The coherence for non-\gls{ris} and \gls{ris} signals are
\begin{subequations}\label{eq:explicit_coherence}
\begin{align}
    C(\bs{g}_l^{\rm n}, \bs{g}^{\rm n}_k) &= \frac{|\langle \bs{a}_{\rm u}(\bs{\theta}_l), \bs{a}_{\rm u}(\bs{\theta}_k)\rangle|^2}{N_{\rm u}^2} \frac{|\langle \bs{d}(\tau_l), \bs{d}(\tau_k)\rangle|^2}{N^2},\label{eq:nonRIS_coherence}\\
    C(\bs{g}_l^{\rm r}, \bs{g}^{\rm r}_k) &= \frac{|\langle \bs{\nu}(\bs{\phi}_l), \bs{\nu}(\bs{\phi}_k)\rangle|^2}{\Vert \bs{\nu}(\bs{\phi}_l)\Vert^2 \Vert \bs{\nu}(\bs{\phi}_k)\Vert^2} \frac{|\langle \bs{d}(\bar{\tau}_l), \bs{d}(\bar{\tau}_k)\rangle|^2}{N^2}.\label{eq:RIS_coherence}
\end{align}
\end{subequations}
For the non-\gls{ris} signal, the coherence depends on the coherence between the delay response vectors and the coherence between the angle response vectors. On the other hand, for the \gls{ris} signal, the coherence depends on the coherence between the delay response vectors and the coherence between the \gls{ris} response vectors. This emphasizes the resolution capabilities of the \gls{ris} signal compared to the non-\gls{ris} signal: if the \gls{ris} response vectors decohere more quickly as the angle separation $\Vert \bs{\phi}_l - \bs{\phi}_k\Vert$ increases, than the array response vectors does as $\Vert \bs{\theta}_l - \bs{\theta}_k\Vert$ increases, then the \gls{ris} signal has improved angular resolution over the non-\gls{ris} signal.

In the case of two targets, the non-centrality parameter $\mu_2$, and in turn ${\rm AUC}_2$, depends only on the coherence and the receive SNR. We show this relation in \cref{subfig:AUC_coherence} and note that given a receive \gls{snr} of $30$ dB we can achieve an \gls{auc} that is nearly $1$ if the coherence is less than $0.9$, while for a receive \gls{snr} of $20$ dB the coherence should nearly be $0$ to achieve an \gls{auc} that is nearly $1$. We can also shed light on the required receive \gls{snr} to realize requirements of false alarm probability and detection probability as a function of the coherence, see \cref{subfig:SNR_coherence}: having a coherence of $0.1$ rather than $0.9$, we can realize the requirements that $p_{\rm fa}=0.001$ and $p_{{\rm d}, 2} = 0.99$, with $10$ dB less \gls{snr}. This clarifies the trade-off between non-\gls{ris} and \gls{ris}, considering that the \gls{ris} signal can have a lower coherence, but with a lower \gls{snr}.
\begin{figure}[t]
    \centering
    \begin{minipage}{0.7\linewidth}
        \centering
        \subfloat[]{\includegraphics[width=\linewidth, page=1]{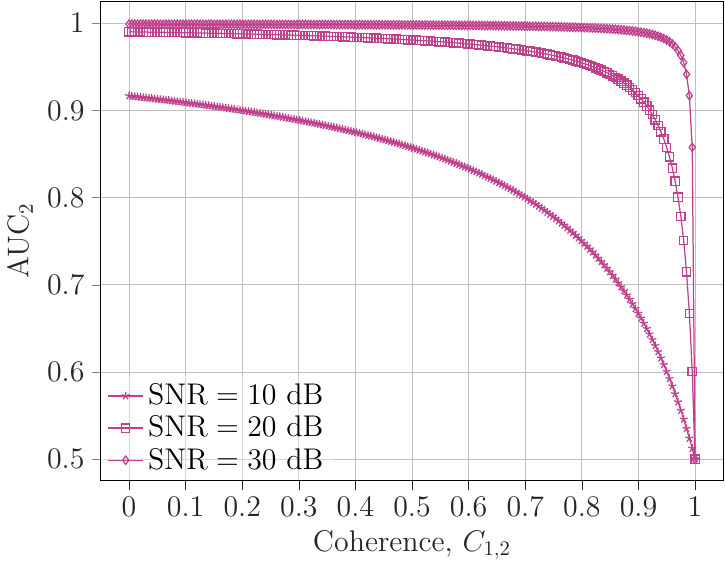}\label{subfig:AUC_coherence}}
    \end{minipage}\\ \vspace{1pt}
    \begin{minipage}{0.7\linewidth}
        \centering
        \subfloat[]{\includegraphics[width=\linewidth, page=2]{figures/P3Figs.pdf}\label{subfig:SNR_coherence}}
    \end{minipage}
    \vspace{1pt}
    \caption{Plot of (a) \gls{auc} against the coherence for varying receive \gls{snr} levels, and (b) the required receive \gls{snr} against coherence to realize false alarm and detection probability requirements, $p_{\rm fa} = 0.001$, $p_{{\rm d}, 2} = 0.9, 0.95, 0.99$.}
    \label{fig:theoretical_coherence_plots}
\end{figure}

We will see in the results that how quickly the \gls{ris} response vectors decohere depends both on the size of the \gls{ris}, $N_{\rm r}$, and the number of \gls{ris} phase profiles, $\Tilde{T}$. Additionally, we note that $\Vert \bs{g}_l^{\rm r}\Vert^2$, and thereby the receive \gls{snr} for the $l$-th path, scales proportionally with the number of \gls{ris} elements, but inverse proportionally to the size of the resolution region in the \gls{ris} angle parameter, i.e., $(\phi^{\text{az}}_{\text{max}} - \phi^{\text{az}}_{\text{min}})(\phi^{\text{el}}_{\text{max}}-\phi^{\text{el}}_{\text{min}})$.

\subsection{Fisher Analysis}\label{subsec:fisher_analysis}
\noindent Using Fisher analysis \cite{AbuShaban2018Bounds}, we will characterize the upper bound on the performance of the channel parameter estimator and position estimator. This is also referred to as the \gls{crlb} and gives us the lower bound of the covariance among unbiased estimators. Specifically, the statement asserts that for an unbiased estimator $\hat{\bs{z}}=\bs{b}(\bs{y})$ of the parameter $\bs{z}$ defined by a function $\bs{b}$ of the data $\bs{y}$, the variance is lower-bounded as $\Var[\hat{\bs{z}}] \geq \bs{F}^{-1}(\bs{z})$, where $\bs{F}(\bs{z})$ is the \gls{fim} and the inequality is in the positive semi-definite sense. Estimators achieving this lower bound are called unbiased and efficient estimators. Such a property can be found, e.g., in the maximum likelihood estimator which is asymptotically unbiased and efficient. \cite{Madsen2010:General}

\subsubsection{General Case}
We begin by deriving the \gls{fim} for the non-\gls{ris} and \gls{ris} signals. Conditioning on the fading parameters $\bs{\alpha}$ the channel is \gls{awgn}, and hence the conditional \gls{fim} for the non-\gls{ris} signal is \cite{AbuShaban2018Bounds}
\begin{subequations}\label{eq:FIM}
\begin{equation}\label{eq:FIM_SB}
    [\bs{F}^{\rm n}]_{i, j} = \frac{4}{\sigma^2}{\rm Re}\left\{\Big(\frac{\partial \bs{\mu}^{\rm n}}{\partial [\Tilde{\bs{\eta}}^{\rm n}]_i}\Big)^H \frac{\partial\bs{\mu}^{\rm n}}{\partial [\Tilde{\bs{\eta}}^{\rm n}]_j}\right\}
\end{equation}
where $\Tilde{\bs{\eta}}^{\rm n} = [\Re\{\bs{\alpha}\}^\top, \Im\{\bs{\alpha}\}^\top, (\bs{\eta}^{\rm n}_1)^\top, \dots, (\bs{\eta}^{\rm n}_L)^\top]^\top$, and $\bs{\mu}^{\rm n} = \bs{G}^{\rm n} \bs{\alpha}$. Conditioning on the fading parameters $\bar{\bs{\alpha}}=[\bar{\alpha}_1,\dots,\bar{\alpha}_L]^\top$, the conditional \gls{fim} for the \gls{ris} signal is
\begin{equation}\label{eq:FIM_DB}
    [\bs{F}^{\rm r}]_{i, j} = \frac{4}{\sigma^2}{\rm Re}\left\{\Big(\frac{\partial \bs{\mu}^{\rm r}}{\partial [\Tilde{\bs{\eta}}^{\rm r}]_i}\Big)^H \frac{\partial\bs{\mu}^{\rm r}}{\partial [\Tilde{\bs{\eta}}^{\rm r}]_j}\right\}
\end{equation}
\end{subequations}
where $\Tilde{\bs{\eta}}^{\rm r} = [\Re\{\bar{\bs{\alpha}}\}^\top, \Im\{\bar{\bs{\alpha}}\}^\top, (\bs{\eta}^{\rm r}_1)^\top, \dots, (\bs{\eta}^{\rm r}_L)^\top]^\top$, and $\bs{\mu}^{\rm r} = \bs{G}^{\rm r} \bar{\bs{\alpha}}$.
The first order derivatives required in the \glspl{fim} can be computed straightforwardly, see Appendix~\ref{subsec:fisher_information_derivations}.

With the \gls{fim} at hand, we can define the \gls{deb} and \gls{aeb} for the relevant parameters:
\begin{subequations}\label{eq:Fisher_bounds}
\begin{align}
    {\rm DEB}_{l}^{\rm n} &= \sqrt{[(\bs{F}^{\rm n})^{-1}]_{2L+l, 2L+l}},\\
    {\rm DEB}_{l}^{\rm r} &= \sqrt{[(\bs{F}^{\rm r})^{-1}]_{4L+l, 4L+l}},\\
    {\rm AEB}_{l}^{\rm n} &= \sqrt{[(\bs{F}^{\rm n})^{-1}]_{3L+l, 3L+l}+[(\bs{F}^{\rm n})^{-1}]_{4L+l, 4L+l}},\\
    {\rm AEB}_{l}^{\rm r} &= \sqrt{[(\bs{F}^{\rm r})^{-1}]_{5L+l, 5L+l}+[(\bs{F}^{\rm r})^{-1}]_{6L+l, 6L+l}}.
\end{align}
We can also characterize the position estimation capabilities with the \gls{sb} and \gls{db} signals through a coordinate system transformation. We define $[\bs{T}^{\rm n}]_{i,j} = \frac{\partial [\bs{\eta}^{\rm n}]_i}{\partial\bs{x}_j}$ and $[\bs{T}^{\rm r}]_{i,j} = \frac{\partial [\bs{\eta}^{\rm r}]_i}{\partial\bs{x}_j}$ as the Jacobian matrices of the transformation from position to channel parameters. The partial derivatives are provided in \cref{eq:Fisher_Jacobian}.
Then, the \glspl{fim} in the Euclidean coordinates are $\bs{F}_{\rm euc}^{\rm n} = (\bs{T}^{\rm n})^\top ([(\bs{F}^{\rm n})^{-1}]_{2L:,2L:})^{-1} \bs{T}^{\rm n}$ and $\bs{F}_{\rm euc}^{\rm r} = (\bs{T}^{\rm r})^\top ([(\bs{F}^{\rm r})^{-1}]_{4L:,4L:})^{-1} \bs{T}^{\rm r}$.\footnote{We treat the $\bs{\theta}_l$ angles as nuisance parameters with the \gls{ris} signal. This is justified by observing that using this signal to estimate the angles of departure at the \gls{ue} when we use just a single precoder will result in a highly inaccurate estimate. Further justification follows from observing through numerical evaluations that the Fisher information of the $\bs{\theta}_l$ parameters is vanishingly small.\label{footnote:theta}} We can define the \gls{peb} as
\begin{align}
    {\rm PEB}^{\rm n}_l &= \sqrt{{\rm Tr}([(\bs{F}_{\rm euc}^{\rm n})^{-1}]_{1+3(l-1):3l,1+3(l-1):3l})},\\
    {\rm PEB}^{\rm r}_l &= \sqrt{{\rm Tr}([(\bs{F}_{\rm euc}^{\rm r})^{-1}]_{1+3(l-1):3l,1+3(l-1):3l})}.
\end{align}
If we assume that the data association between the \gls{sb} and \gls{db} estimator results is known, then the joint \gls{fim} can be formulated as $\bs{F} = {\rm bdiag}(\bs{F}^{\rm n}, \bs{F}^{\rm r})$ where ${\rm bdiag}$ denotes block-diagonal matrix,
and as a consequence, in Euclidean coordinates the non-\gls{ris} and \gls{ris} Fisher information is simply aggregated: $\bs{F}_{\rm euc} = \bs{F}_{\rm euc}^{\rm n} + \bs{F}_{\rm euc}^{\rm r}$.
The \gls{peb} for the joint sensing is
\begin{equation}\label{eq:peb}
    {\rm PEB}_l = \sqrt{{\rm Tr}([\bs{F}^{-1}_{\rm euc}]_{1+3(l-1):3l,1+3(l-1):3l})}.
\end{equation}
\end{subequations}
We can compute explicit expressions for the equivalent Fisher information for individual parameters. Full derivations are cumbersome, however, we can gain insight into the full problem by considering a simplified but sufficiently complicated special case.

\subsubsection{Case study: Special case of two targets}


Consider the special case of two targets, i.e., $L=2$. Beginning with the non-\gls{ris} signal, we assume for simplicity that $\theta_1^{\rm el}$, $\theta_2^{\rm el}$, $\tau_1$, and $\tau_2$ are known, such that the unknowns are $\Re\{\alpha_1\}$, $\Re\{\alpha_2\}$, $\Im\{\alpha_1\}$, $\Im\{\alpha_2\}$, $\theta_1^{\rm az}$, and $\theta_2^{\rm az}$. We can now focus on the equivalent Fisher information for $\theta_1^{\rm az}$ which provides the necessary understanding to interpret the information loss due to the unknown channel coefficients and the interfering second target:
%
\begin{equation}\label{eq:fisher_finaleq}
    F({\rm n}, \theta_1^{\rm az}) := \frac{4}{\sigma^2} \bigg(|\alpha_1|^2 \Vert \bs{g}_{\theta_1^{\rm az}}^{\rm n}\Vert^2 \Big(1 - {\rm IL}_{\bs{\alpha}}\Big) - {\rm IL}_{{\theta}_2^{\rm az}}\bigg),
\end{equation}
where ${\rm IL}_{\bs{\alpha}}$ is the information loss due to the lack of knowledge about the channel coefficients
\begin{equation*}
    {\rm IL}_{\bs{\alpha}} := \frac{\Vert \bs{x}_{\theta_1^{\rm az}}^{\rm n} \Vert^2}{\Vert \bs{g}_1^{\rm n}\Vert^2 \Vert \bs{g}_2^{\rm n}\Vert^2 \Vert \bs{g}_{\theta_1^{\rm az}}^{\rm n}\Vert^2 (1 - C_{1,2}^{\rm n})} = \Tilde{C}(\bs{g}_1^{\rm n}, \bs{g}_2^{\rm n}, \bs{g}_{\theta_1^{\rm az}}^{\rm n}),
\end{equation*}
and ${\rm IL}_{{\theta}_2^{\rm az}}$ is the information loss due to the lack of knowledge of the azimuth angle of the second target
\begin{equation*}
    {\rm IL}_{{\theta}_2^{\rm az}} := \frac{\Re\Big\{\alpha_1^*\alpha_2 \Big(\langle \bs{g}_{\theta_1^{\rm az}}^{\rm n}, \bs{g}_{\theta_2^{\rm az}}^{\rm n}\rangle - \frac{\langle \bs{x}_{\theta_1^{\rm az}}^{\rm n}, \bs{x}_{\theta_2^{\rm az}}^{\rm n} \rangle}{\Vert \bs{g}_1^{\rm n}\Vert^2 \Vert \bs{g}_2^{\rm n}\Vert^2 (1 - C_{1,2}^{\rm n})}\Big)\Big\}^2}{|\alpha_2|^2 \Vert \bs{g}_{\theta_2^{\rm az}}^{\rm n}\Vert^2 \Big(1 - \Tilde{C}(\bs{g}_1^{\rm n}, \bs{g}_2^{\rm n}, \bs{g}_{\theta_2^{\rm az}}^{\rm n})\Big)},
\end{equation*}
where we have defined
$\bs{x}_{\theta_1^{\rm az}}^{\rm n} := \bs{g}_2^{\rm n} \langle \bs{g}_1^{\rm n}, \bs{g}_{\theta_1^{\rm az}}^{\rm n}\rangle - \bs{g}_1^{\rm n} \langle \bs{g}_2^{\rm n}, \bs{g}_{\theta_1^{\rm az}}^{\rm n}\rangle$,    $\bs{x}_{\theta_2^{\rm az}}^{\rm n} := \bs{g}_2^{\rm n} \langle \bs{g}_1^{\rm n}, \bs{g}_{\theta_2^{\rm az}}^{\rm n}\rangle - \bs{g}_1^{\rm n} \langle \bs{g}_2^{\rm n}, \bs{g}_{\theta_2^{\rm az}}^{\rm n}\rangle$,
$\bs{g}_{\theta_1^{\rm az}}^{\rm n} := \frac{\partial \bs{g}_1^{\rm n}}{\partial \theta_1^{\rm az}}$, and $\bs{g}_{\theta_2^{\rm az}}^{\rm n} := \frac{\partial \bs{g}_2^{\rm n}}{\partial \theta_2^{\rm az}}$ (see Appendix~\ref{subsec:fisher_information_derivations}).

We observe that if $\theta_1^{\rm az} = \theta_2^{\rm az}$, then $F({\rm n}, \theta_1^{\rm az}) = 0$. On the other hand, if $C(\bs{g}_1^{\rm n}, \bs{g}_2^{\rm n})$, $C(\bs{g}_1^{\rm n}, \bs{g}_{\theta_2^{\rm az}}^{\rm n})$, $C(\bs{g}_2^{\rm n}, \bs{g}_{\theta_1^{\rm az}}^{\rm n})$, and $C(\bs{g}_{\theta_1^{\rm az}}^{\rm n}, \bs{g}_{\theta_2^{\rm az}}^{\rm n})$
are all equal to zero, then the equivalent Fisher information is maximized as $F({\rm n}, \theta_1^{\rm az}) = \frac{4|\alpha_1|^2}{\sigma^2} \Vert \bs{g}_{\theta_1^{\rm az}}^{\rm n}\Vert^2 \big(1 - C(\bs{g}_1^{\rm n}, \bs{g}_{\theta_1^{\rm az}}^{\rm n})\big)$, i.e., there is no loss of information of the first target due to the second target. Inspecting further these inner products, it is derived that this condition comes down to $C(\bs{a}_{\rm u}(\bs{\theta}_1), \bs{a}_{\rm u}(\bs{\theta}_2))$, $C(\bs{a}_{\rm u}(\bs{\theta}_1), \frac{\partial\bs{a}_{\rm u}(\bs{\theta}_2)}{\partial \theta_2^{\rm az}})$, $C(\bs{a}_{\rm u}(\bs{\theta}_2), \frac{\partial\bs{a}_{\rm u}(\bs{\theta}_1)}{\partial \theta_1^{\rm az}})$, and $C(\frac{\partial \bs{a}_{\rm u}(\bs{\theta}_1)}{\partial \theta_1^{\rm az}}, \frac{\partial\bs{a}_{\rm u}(\bs{\theta}_2)}{\partial \theta_2^{\rm az}})$
being equal to zero. Generally speaking these coherences decrease as $|\theta_2^{\rm az}-\theta_1^{\rm az}|$ increases, and the rate at which this decrease occurs depends on $N_{\rm u}^{\rm az}$ \cite{Emadi2018:OMP}.

Considering now the \gls{ris} signal and assuming  that $\theta_1^{\rm az}$, $\theta_2^{\rm az}$, $\theta_1^{\rm el}$, $\theta_2^{\rm el}$, $\tau_1$, $\tau_2$, $\phi_1^{\rm el}$, and $\phi_2^{\rm el}$ are known, and the unknowns are $\Re\{\alpha_1\}$, $\Re\{\alpha_2\}$, $\Im\{\alpha_1\}$, $\Im\{\alpha_2\}$, $\phi_1^{\rm az}$, and $\phi_2^{\rm az}$. Reusing again \cref{eq:fisher_finaleq}, we find that the equivalent Fisher information for $\phi_1^{\rm az}$ with the \gls{ris} signal is $F({\rm r}, \phi_1^{\rm az})$, replacing in \cref{eq:fisher_finaleq} all superscript ${\rm n}$ by ${\rm r}$, all $\theta_1^{\rm az}$ by $\phi_1^{\rm az}$, and all $\alpha_1,\alpha_2$ by $\bar{\alpha}_1,\bar{\alpha}_2$. As before, we observe that when $\phi_1^{\rm az} = \phi_2^{\rm az}$, then $F({\rm r}, \phi_1^{\rm az}) = 0$, but if $C(\bs{\nu}(\bs{\phi}_1), \bs{\nu}(\bs{\phi}_2))$, $C(\bs{\nu}(\bs{\phi}_1), \frac{\partial\bs{\nu}(\bs{\phi}_2)}{\partial \phi_2^{\rm az}})$, $C(\bs{\nu}(\bs{\phi}_2), \frac{\partial\bs{\nu}(\bs{\phi}_1)}{\partial \phi_1^{\rm az}})$, and $C(\frac{\partial \bs{\nu}(\bs{\phi}_1)}{\partial \phi_1^{\rm az}}, \frac{\partial\bs{\nu}(\bs{\phi}_2)}{\partial \phi_2^{\rm az}})$
are all equal to zero, then the Fisher information is maximized.

\begin{remark}
We have seen that the coherence plays a central role for both the \gls{crlb} and the detection probability. For the detection probability, it is the coherences between vectors $\bs{g}_l$, $l=1,\dots,L$, which appear, while for the \gls{crlb} it is also the coherences between vectors $\bs{g}_l$ and their partial derivatives.
\end{remark}

\section{Channel and Position Estimation}\label{sec:channel_estimation}
\noindent With the formulations \cref{eq:nonris_lin_mod,eq:ris_lin_mod}, we can in parallel apply channel estimation techniques such as \gls{em} \gls{ml} \cite{Spagnolini2018:Statistical}, compressive sensing \cite{Malioutov2005:Sparse,Tropp2007:OMP}, or beamspace \gls{music} \cite{Guo2017Beamspace}, followed by a data association step, and subsequent position estimation.
We begin by briefly recalling the well-studied channel parameter estimation technique \gls{omp}

\subsection{Channel parameter estimation}

\subsubsection{OMP}
Consider an observed signal $\bs{y} = \bs{G} \bs{\alpha} + \bs{\varepsilon}$, which can be either the non-\gls{ris} or \gls{ris} signal \cref{eq:nonris_lin_mod,eq:ris_lin_mod}. To apply the \gls{omp} method, we construct a redundant dictionary $\bs{\Psi} = [\bs{g}_1, \dots, \bs{g}_M]$ where $M$ is a redundancy parameter, and $\bs{g}_i = \bs{g}^{\rm n}(\bs{\eta}^{\rm n}_i)$ with $\bs{\eta}^{\rm n}_i$ in the resolution region for the non-\gls{ris} signal, while $\bs{g}_i = \bs{g}^{\rm r}(\bs{\eta}^{\rm r}_i)$ with $\bs{\eta}^{\rm r}_i$ in the resolution region for the \gls{ris} signal. Then, we search for a sparse representation $\bs{v} \in \C^{M}$ such that $\Vert \bs{y} - \bs{\Psi} \bs{v}\Vert_2$ is minimized by initializing $\bs{v}$ as the zero vector and then iteratively updating the support set. Specifically, the iterative update is found by maximizing the \gls{omp} objective $|\langle \bs{g}_i, \bs{r} \rangle|$ with respect to index $i$, where $\bs{r}$ is the residual.
The algorithm is summarized below \cite{Tropp2007:OMP}:
\begin{itemize}
    \item Let $\bs{r} = \bs{y}$ denote the residual, initialize $\bs{v} = \bs{0}$, and let $\Lambda = \emptyset$ denote the support set of $\bs{v}$.
    \item While $\Vert \bs{r} \Vert_2 \geq {\rm residual}_{\rm th}$ do:
    \begin{enumerate}
        \item $I \gets \argmax_{i=1,\dots,M} |\langle \bs{g}_i, \bs{r}\rangle|$.
        \item $\Lambda \gets \Lambda \cup \{I\}$.
        \item $[\bs{v}]_\Lambda \gets [\bs{\Psi}]_\Lambda^\dagger \bs{y}$.~\footnote{The notation $[\bs{\Psi}]_\Lambda$ refers to the sub-matrix of $\bs{\Psi}$ containing the columns with indices in the index set $\Lambda$. Similarly, $[\bs{v}]_\Lambda$ is the sub-vector of $\bs{v}$ containing the elements with indices in the index set $\Lambda$.}
        \item $\bs{r} \gets \bs{y} - [\bs{\Psi}]_\Lambda [\bs{v}]_\Lambda$.
    \end{enumerate}
\end{itemize}
The support set $\Lambda$ specifies a set of estimated channel parameters $\{\bs{\eta}^{\rm n}_i\}_{i\in\Lambda}$ or $\{\bs{\eta}^{\rm r}_i\}_{i\in\Lambda}$ in case of non-\gls{ris} or \gls{ris}, respectively. The threshold parameter ${\rm residual}_{\rm th} > 0$ specifies the trade-off between detection probability and false alarm probability, and can be chosen based on the noise variance, specifically ${\rm residual}^2_{\rm th}$ should be larger than $\sigma^2N \Tilde{T} N_{\rm u}/2$ such that the algorithm terminates when the residual only contains noise. We let ${\rm residual}_{\rm th}^{\rm n}$ and ${\rm residual}_{\rm th}^{\rm r}$ denote the thresholds for the non-\gls{ris} and \gls{ris} signals, respectively.

\subsection{Data Association}
\noindent Once we have estimated the non-\gls{ris} and \gls{ris} channel parameters,
the data association task is to find pairs of non-\gls{ris} and \gls{ris} channel parameters which associate with the same target. We can compare the measurements in the domain of the \gls{aoa} at the \gls{ris}, by noting that the non-\gls{ris} parameters uniquely identify a position which subsequently uniquely identifies the \gls{ris} channel parameters. Specifically, from the non-\gls{ris} parameter estimate $\hat{\bs{\eta}}_i^{\rm n}$ we can uniquely find the estimate of the position of the target $\hat{\bs{c}}_i^{\rm n} = \bs{h}_{\rm n}^{-1}(\hat{\bs{\eta}}_i^{\rm n})$ from which we can compute the \gls{aoa} at the \gls{ris} as $\hat{\bs{\phi}}_i^{\rm n} = \bs{\phi}(\hat{\bs{c}}_i^{\rm n})$, cf. Appendix~\ref{subsec:channel_parameters}.

With this in mind, to solve the data association problem, we define the cost matrix
\begin{equation}
    \begin{split}
        \bs{C}_{i, k} = &\min \big( ( [\hat{\bs{\phi}}_i^{\rm n}]_1 - [\hat{\bs{\phi}}_k]_1 )^2, (\lvert[\hat{\bs{\phi}}_i^{\rm n}]_1 - [\hat{\bs{\phi}}_k]_1\rvert - 2\pi )^2 \big)\\
        &+ \min \big( ( [\hat{\bs{\phi}}_i^{\rm n}]_2 - [\hat{\bs{\phi}}_k]_2 )^2, (\lvert[\hat{\bs{\phi}}_i^{\rm n}]_2 - [\hat{\bs{\phi}}_k]_2\rvert - \pi )^2 \big),
    \end{split}
\end{equation}
where $i=1,\dots,\hat{L}_{\rm n}$ and $k=1,\dots,\hat{L}_{\rm r}$, $\hat{L}_{\rm n}$ and $\hat{L}_{\rm r}$ are the estimated number of targets with the non-\gls{ris} and \gls{ris} signals, respectively, and $\hat{\bs{\phi}}_k$ are the estimated \glspl{aoa} at the \gls{ris} with the \gls{ris} signal. Let $\bs{X}\in\{0,1\}^{\hat{L}_{\rm n} \times \hat{L}_{\rm r}}$ be the boolean assignment matrix such that $\bs{X}_{i,k} = 1$ if the $i$-th non-\gls{ris} estimate is associated with the $k$-th \gls{ris} estimate and $\bs{X}_{i,k} = 0$ otherwise. When $\hat{L}_{\rm r} \geq \hat{L}_{\rm n}$, to solve the data association problem is to solve the optimization problem
\begin{align}
    \operatorname{minimize}~~&~ \text{Tr}(\bs{X}^\top \bs{C})\\
    \text{subject to}~~&~\bs{X}_{i,k}\in\{0,1\},~\forall i,k,\nonumber\\
    &\sum_i \bs{X}_{i,k} \leq 1,~\forall k, \:~~ \sum_k \bs{X}_{i,k} = 1,~\forall i. \notag
\end{align}
The constraints mean that each \gls{ris} detection is associated with only one non-\gls{ris} detection and vice versa. Now, if more detection are made with the \gls{ris} signal (resp. non-\gls{ris} signal), the targets that are not associated with a non-\gls{ris} detection (resp. \gls{ris} detection), are kept and can be used for subsequent positioning, although with less precision.
The estimated number of targets is $\hat{L} = {\rm max}(\hat{L}_{\rm n}, \hat{L}_{\rm r})$.

\subsection{Position Estimation}
\noindent By the definition of the channel parameters, as detailed in Appendix~\ref{subsec:channel_parameters}, the mapping $\bs{h}$ from Euclidean space to the channel parameter space is known.
We formulate the weighted non-linear least squares optimization problem as
\begin{equation}\label{eq:position_optimization}
    \hat{\bs{c}} = \argmin_{\bs{c}\in\mathcal{D}}~~\frac{1}{2} (\bs{h}(\bs{c}) - \hat{\bs{\eta}})^\top \bs{F}_{\rm efim} (\bs{h}(\bs{c}) - \hat{\bs{\eta}}),
\end{equation}
where \resizebox{0.88\linewidth}{!}{$\bs{F}_{\rm efim} = {\rm bdiag}(([(\bs{F}^{\rm n})^{-1}]_{2L:,2L:})^{-1}, ([(\bs{F}^{\rm r})^{-1}]_{4L:,4L:})^{-1})$} is the \gls{efim} evaluated at $\hat{\bs{\eta}}=[\hat{\tau}, \hat{\bar{\tau}}, \hat{\bs{\theta}}^\top, \hat{\bs{\phi}}^\top]^\top$.
The partial derivatives required in the Jacobian is provided in \cref{eq:Fisher_Jacobian}, and any gradient-based minimization method can solve this numerically.



\section{Numerical Experiments}\label{sec:numerical_experiments}
\noindent In this section, the high resolution sensing capabilities of \gls{ris} is evaluated through numerical experiments. First we specify the setting, followed by a presentation of the key performance indicators. Then, the working principle is illustrated for an exemplary scenario, followed by a study of the coherence and the theoretical bounds using the results from \cref{sec:theoretical_analysis}. Finally, the performance of the sensing algorithm outlined in \cref{sec:channel_estimation} with and without the \gls{ris} is analyzed.

\subsection{Settings}\label{subsec:settings}
\begin{table}[t]
    \centering
    \caption{Scenario \& hyperparameter settings}\label{tab:params}
    \begin{tabular}{ccc}
        \toprule
        \textbf{Parameter}	& \textbf{Value} & \textbf{Description} \\ \midrule
        $\bs{s}$ & [0, 0, 0, 0, 0, 0] & \gls{ue} state\\
        $\bs{s}_{\rm r}$ & [3, 5, 6, 0, 0, 0] & \gls{ris} state\\
        $f_c$ & 15 GHz & Carrier frequency\\
        $\lambda$ & 2 cm & Wavelength\\
        $W$ & 9 MHz & Bandwidth\\
        $\Delta_f$ & 120 kHz & Subcarrier spacing\\
        $N$ & 75 & Number of subcarriers\\
        $E_{\rm s} W T N_{\rm u}$ & 65 dBm & Total transmission energy\\
        $N_0$ & $-166$ dBm & Noise power spectral density\\
        $N_{\rm r}$ & 1225 $(35 \times 35)$ & RIS array size\\
        $N_{\rm u}$ & 4 $(2 \times 2)$ & UE array size\\
        $T$ & 50 & \gls{ofdm} symbols\\
        ${\rm ref}_{\rm th}$ & $0.3$ & Reference threshold\\ 
        \bottomrule
        \vspace{1pt}
    \end{tabular}
\end{table}
\begin{figure}[t]
    \centering
    \begin{minipage}{0.23\textwidth}
        \subfloat[Non-\gls{ris} iteration 1.]{\includegraphics[width=\linewidth, page=3]{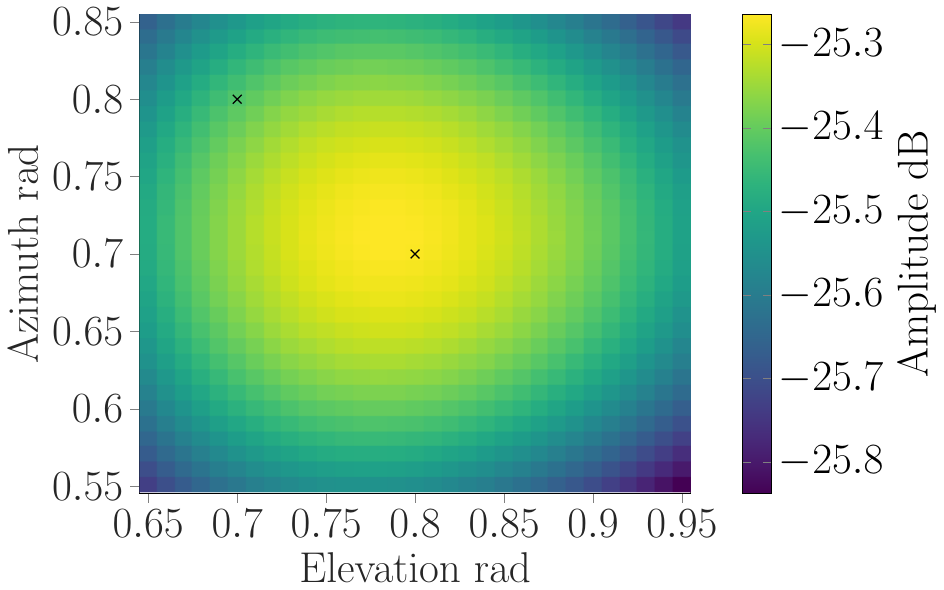}\label{subfig:nonRIS1_el_az}}
    \end{minipage}
    ~
    \begin{minipage}{0.23\textwidth}
        \subfloat[Non-\gls{ris} iteration 2.]{\includegraphics[width=\linewidth, page=4]{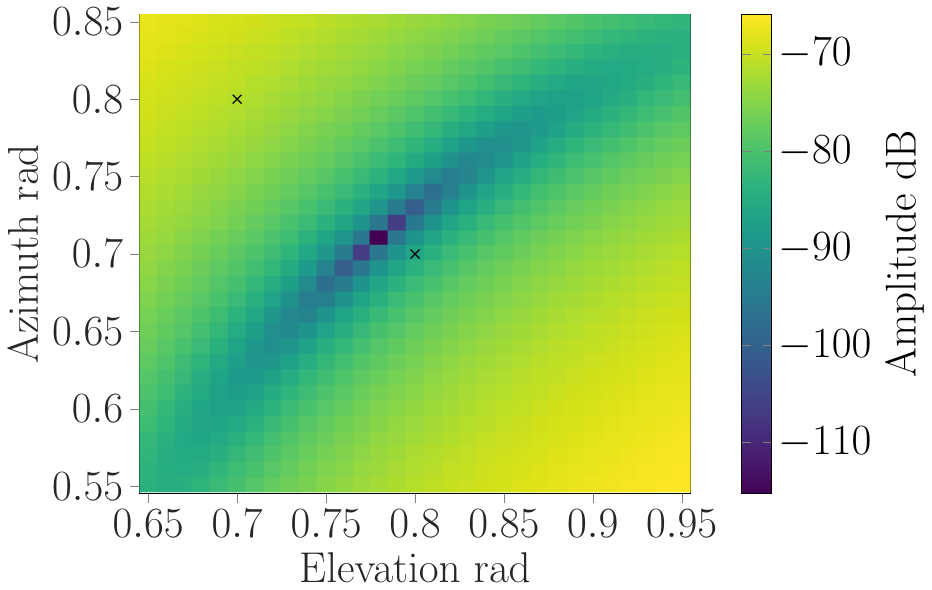}\label{subfig:nonRIS2_el_az}}
    \end{minipage}
    \\\vspace{1pt}
    \begin{minipage}{0.23\textwidth}
        \subfloat[\gls{ris} iteration 1.]{\includegraphics[width=\linewidth, page=5]{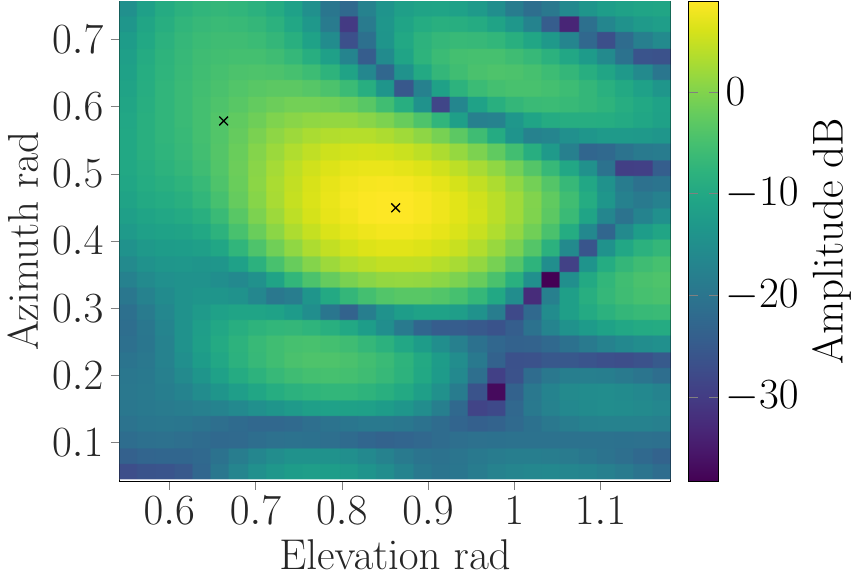}\label{subfig:RIS1_el_az}}
    \end{minipage}
    ~
    \begin{minipage}{0.23\textwidth}
        \subfloat[\gls{ris} iteration 2.]{\includegraphics[width=\linewidth, page=6]{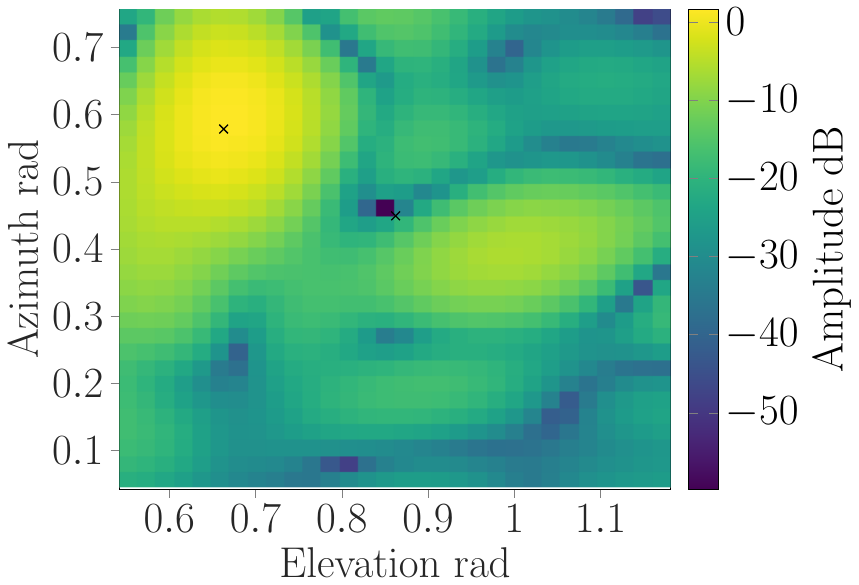}\label{subfig:RIS2_el_az}}
    \end{minipage}
    \vspace{1pt}
    \caption{Non-\gls{ris} and \gls{ris} signal \gls{omp} objective, $|\langle \bs{g}_i, \bs{r}\rangle|$. The black crosses indicate the true positions of the \glspl{sp}.}
    \label{fig:OMP}
\end{figure}
\begin{figure}[t]
    \centering
    \begin{minipage}{0.7\linewidth}
        \centering
        \includegraphics[width=\linewidth, page=7]{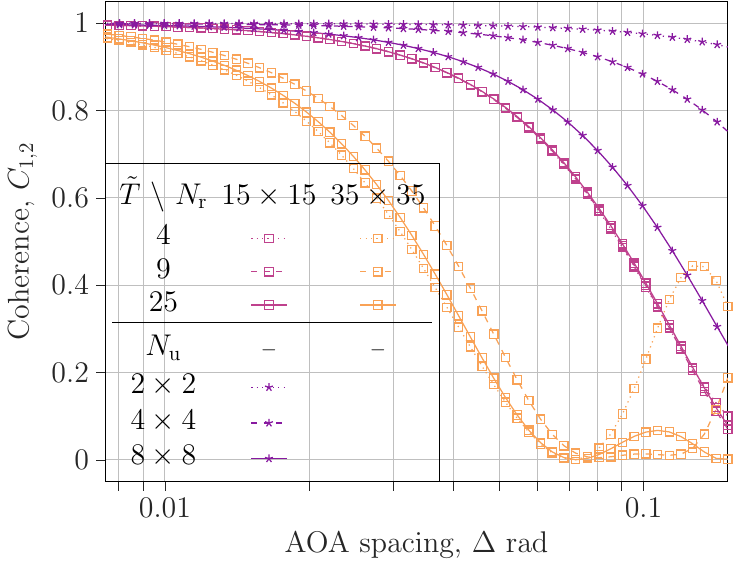}
    \end{minipage}
    \caption{Plot of coherence against the spacing of the targets in the \gls{aoa} at the \gls{ue} for varying \gls{ris} specifications, $\Tilde{T}$ and $N_{\rm r}$, and \gls{ue} specifications, $N_{\rm u}$.}
    \label{fig:coherence_plots}
\end{figure}
\begin{figure}[t]
    \centering
    \includegraphics[width=0.7\linewidth, page=8]{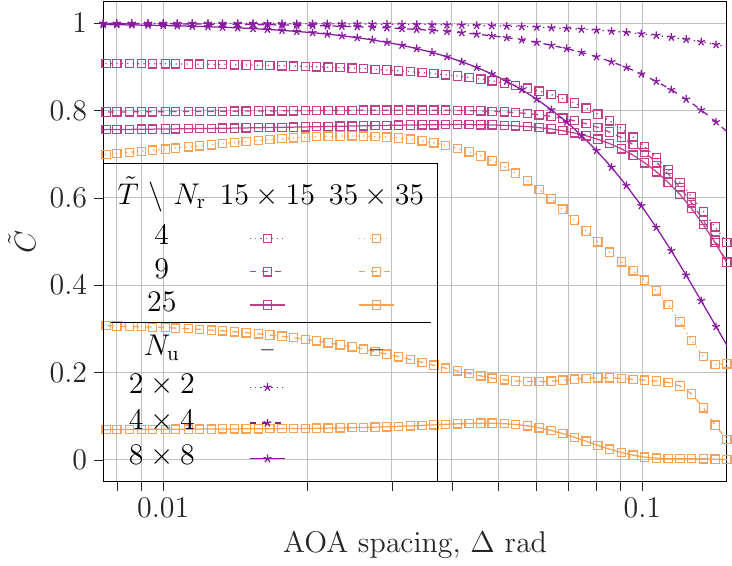}
    \caption{Plot of $\Tilde{C}$ against the spacing between the targets in the \gls{aoa} at the \gls{ue}.}
    \label{fig:TildeC}
\end{figure}
\begin{figure}[t]
    \centering
    \includegraphics[width=0.7\linewidth, page=9]{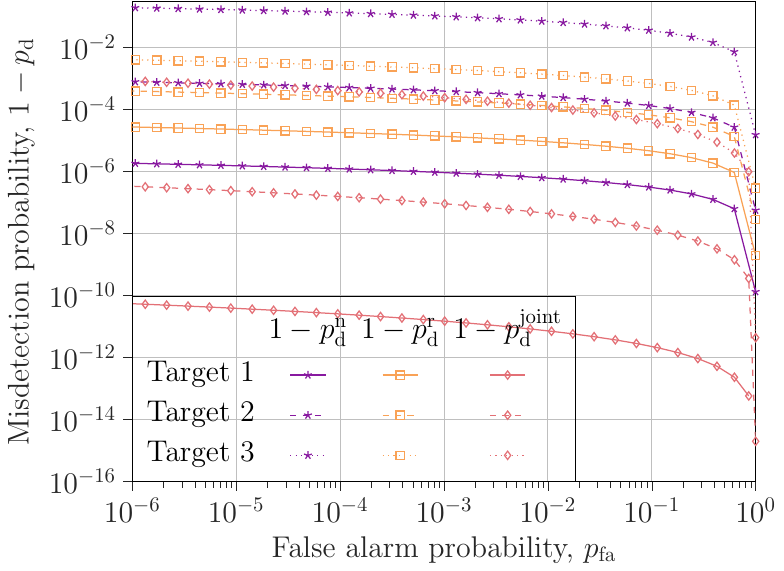}
    \caption{Expected misdetection probability against the false alarm probability on a log-log scale.}
    \label{fig:Epmd_ROC}
\end{figure}

\noindent In the simulations, we let the bandwidth be only $W=9~\rm{MHz}$ and the number of \gls{ue} antennas be just $N_{\rm u}=2\times 2$. These settings are chosen as we are interested in \glspl{ue} with limited sensing resolution. The transmission energy is set to $E_{\rm s} W T N_{\rm u} = 65~\rm{dBm}$, $E_{\rm s}$ is the energy per subcarrier, cf.~Appendix~\ref{subsec:channel_parameters}, with a noise power spectral density of $N_0=-166~\rm{dBm}$, $\sigma^2 = N_0 W$, consistent with the specifications used in \cite{Hyowon2023:Bounce}.
The size of the \gls{ris} is set to $N_{\rm r}=35\times 35$ with number of \gls{ofdm} symbols set to $T=50$ which results in $\Tilde{T}=25$ different \gls{ris} configurations in the beam sweep.
Settings for the numerical experiments are summarized in \cref{tab:params}.

We will consider an \gls{sp} cluster of up to three targets with channel parameters $\bs{\eta}_1^{\rm n} = [\tau_1, \theta_1^{\rm az}, \theta_1^{\rm el}] = [120, 0.7, 0.8]$, $\bs{\eta}_1^{\rm n} = [\tau_1, \theta_1^{\rm az} + \Delta, \theta_1^{\rm el} - \Delta]$, and $\bs{\eta}_3^{\rm n} = [\tau_1, \theta_1^{\rm az} - \Delta, \theta_1^{\rm el} + \Delta]$ where the delay is expressed in nanoseconds (ns), the angles are expressed in radians (rad), and $\Delta > 0$ rad is the spacing between the targets in the \gls{aoa} at the \gls{ue}.
The expected radar cross sections are $\sigma^2_{\rm rcs, 1} = 50~{\rm m}^2$, $\sigma^2_{\rm rcs, 2} = 5.0~{\rm m}^2$, and $\sigma^2_{\rm rcs, 3} = 0.5~{\rm m}^2$, cf.~Appendix~\ref{subsec:channel_parameters}. Hence, on average, the signal received from target one is the strongest target with significantly weaker signals from targets two and three.

\subsection{Performance Metrics}\label{subsec:performance_metrics}
\noindent We evaluate the performance with the empirical expected \gls{auc} and the \gls{gospa} metric, see \cite{Rahmathullah2017:GOSPA} for details. For the \gls{gospa} metric, we use the $2$-norm as the pairwise distance, let the maximum allowable localization error be fixed as $5$ m, and choose the cardinality penalty normalization of $2$ as recommended \cite{Rahmathullah2017:GOSPA}.
We use ${\rm OSPA}^{\rm n}$, ${\rm OSPA}^{\rm r}$, and ${\rm OSPA}^{\rm joint}$ to denote the \gls{ospa} metric using the non-\gls{ris} signal, the \gls{ris} signal, and the combination of the two, respectively.

We consider, the probability of detecting $l \leq L$ targets conditioned on the channel parameters defined as $p_{{\rm d}, l} = \E[\mathbbm{1}[\hat{L} \geq l]]$, with corresponding probability of false alarm $p_{\rm fa} = \E[\mathbbm{1}[\hat{L} > L]]$, where $\mathbbm{1}[x]$ is the indicator function equal to one when the condition $x$ is true and equal to zero when $x$ is false.
The empirical detection probability, denoted $\hat{p}_{{\rm d}, l}$, and empirical false alarm probability, denoted $\hat{p}_{\rm fa}$,
replaces the expectation operators with averages over Monte Carlo simulations. The empirical expected \gls{auc}, denoted $\widehat{{\rm AUC}}_{l}$ is found by integrating $\hat{p}_{{\rm d}, l}$ considered a function of $\hat{p}_{\rm fa}$. We add superscripts ${\rm n}$, ${\rm r}$, and ${\rm joint}$ to denote the empirical results using the non-\gls{ris}, \gls{ris}, and the combination of the two, respectively.

\begin{figure*}[t]
    \centering
    \begin{minipage}{0.3\textwidth}
        \centering
        \subfloat[Target 1.]{\includegraphics[width=\linewidth, page=10]{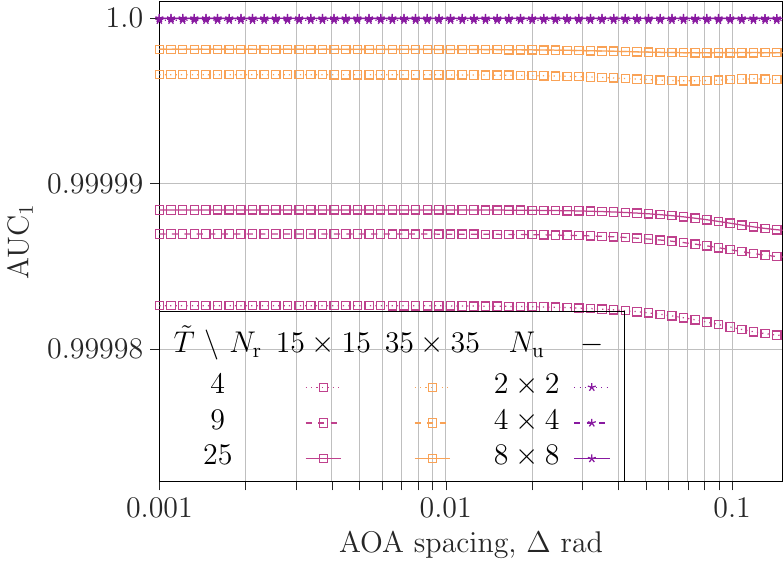}\label{subfig:AUC_Epd_spacing_Target1}}
    \end{minipage}
    ~
    \begin{minipage}{0.3\textwidth}
        \centering
        \subfloat[Target 2.]{\includegraphics[width=\linewidth, page=11]{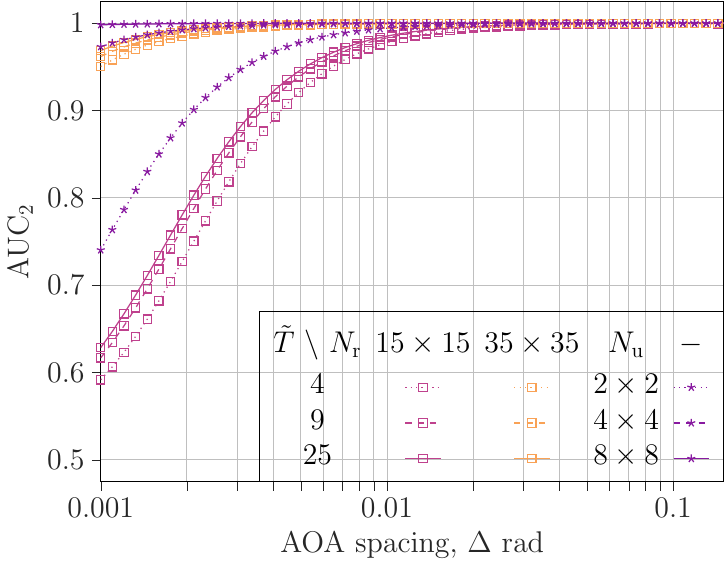}\label{subfig:AUC_Epd_spacing_Target2}}
    \end{minipage}
    ~
    \begin{minipage}{0.3\textwidth}
        \centering
        \subfloat[Target 3.]{\includegraphics[width=\linewidth, page=12]{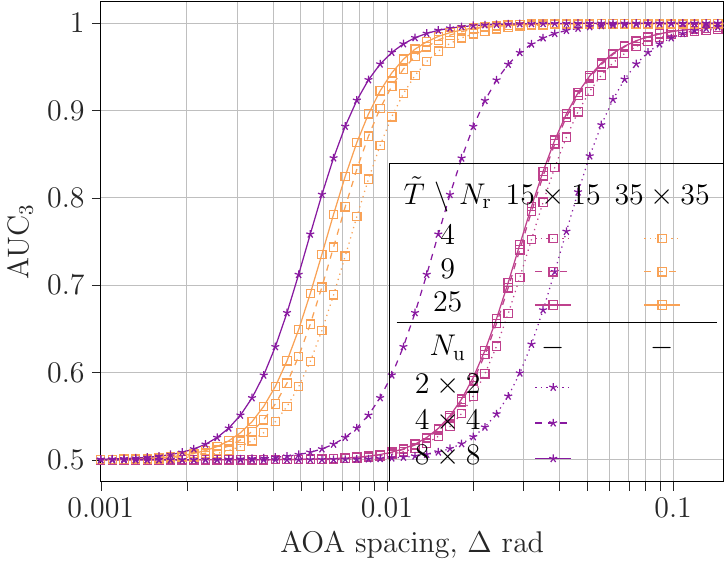}\label{subfig:AUC_Epd_spacing_Target3}}
    \end{minipage}
    \caption{
    Plot of the expected \gls{auc} dependent on the spacing between the targets in the \gls{aoa} at the \gls{ue} for varying number of \gls{ris} elements, $N_{\rm r}$, \gls{ris} phase profiles, $\Tilde{T}$, and number of \gls{ue} antennas, $N_{\rm u}$.
    }
    \label{fig:Epd_spacing}
\end{figure*}

\begin{figure*}[t]
    \centering
    \begin{minipage}{0.3\textwidth}
        \centering
        \subfloat[Target 1.]{\includegraphics[width=\linewidth, page=13]{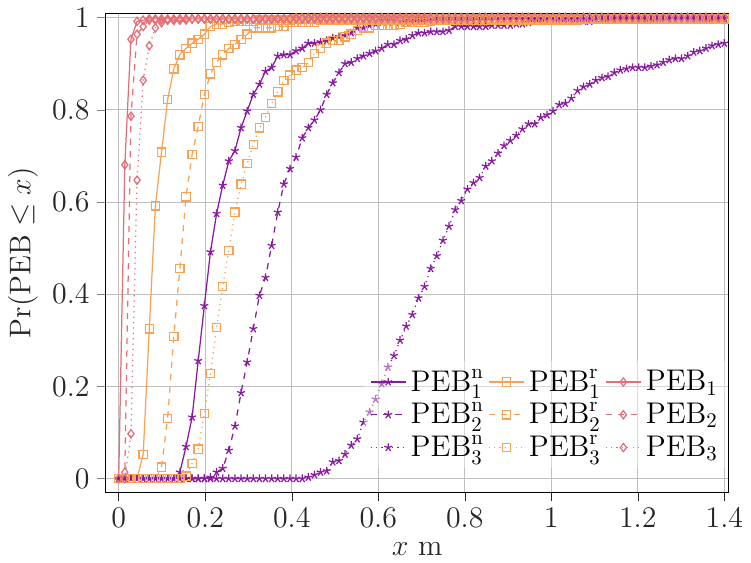}\label{subfig:PEBcdf}}
    \end{minipage}
    ~
    \begin{minipage}{0.3\textwidth}
        \centering
        \subfloat[Target 2.]{\includegraphics[width=\linewidth, page=14]{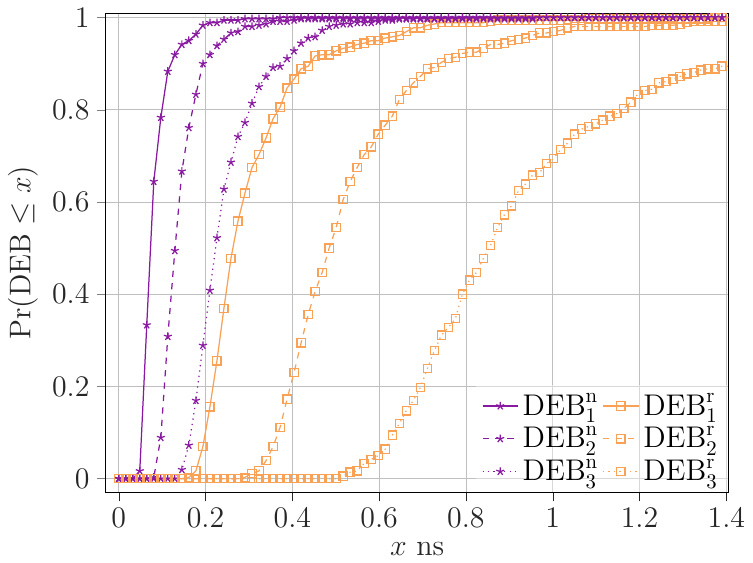}\label{subfig:DEBcdf}}
    \end{minipage}
    ~
    \begin{minipage}{0.3\textwidth}
        \centering
        \subfloat[Target 3.]{\includegraphics[width=\linewidth, page=15]{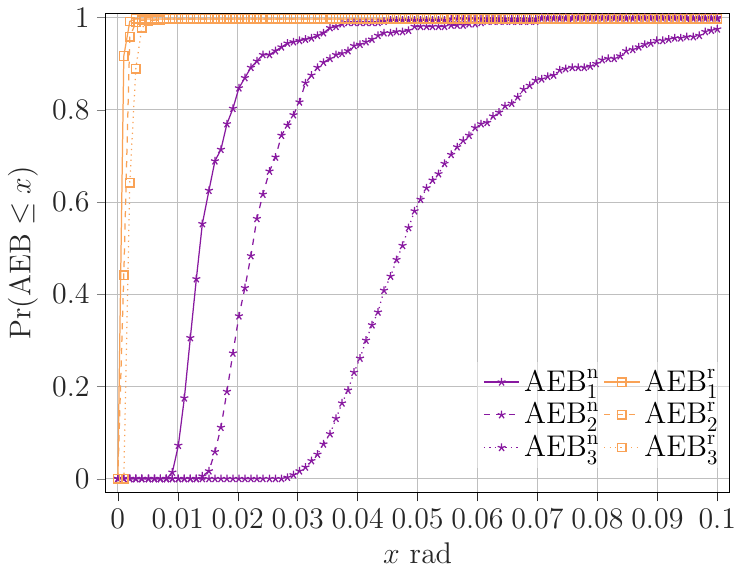}\label{subfig:AEBcdf}}
    \end{minipage}
    \caption{\glspl{cdf} of the \gls{peb}, \gls{deb}, and \gls{aeb} for the three targets with either the non-\gls{ris}, \gls{ris}, or combined. The spacing between the targets is $\Delta=0.1$.}
    \label{fig:FisherCDF}
\end{figure*}

\subsection{Working Principle}
\noindent To illustrate the working principle, we include only the first two targets with $\Delta = 0.07$ rad.
This corresponds to \glspl{sp} located at $[9.87, 8.31, 12.53]~{\rm m}$ and $[8.07, 8.31, 13.76]~{\rm m}$ meaning that the distance between the two \glspl{sp} is $2.17~{\rm m}$.
In this situation, the channel parameters in the view of the \gls{ue} varies only slightly along the azimuth and elevation angles, and it will be difficult for the \gls{ue} alone to detect both \glspl{sp}.

In \cref{fig:OMP}, the \gls{omp} objective, $|\langle \bs{g}_i, \bs{r}\rangle|$, for the non-\gls{ris} and \gls{ris} signals are shown\footnote{For visualization purposes, we evaluate the \gls{omp} objective in the true delay, and display the \gls{omp} objective only as a function of azimuth and elevation.}. From \cref{subfig:nonRIS2_el_az} we see that with the non-\gls{ris} signal we cannot clearly resolve the second \gls{sp} as we cannot see a distinguishable peaks in the second iteration, however, using the \gls{ris} signal we see that the angles are well-separated in the view of the \gls{ris}, and we notice a distinguishable peak in \cref{subfig:RIS2_el_az}.

\subsection{Coherence and Detection Bound}
\noindent To use the results on the expected detection probability \cref{corr:expected_detection_probability,prop:expected_joint_detection_probability}, we require that $\mathcal{R}_o$ is fixed. Let $\mathcal{R}_o = [112, 128] \times [0.55, 0.85] \times [0.65, 0.95]~{\rm ns} \cdot {\rm rad}^2$ be the fixed resolution region.

\subsubsection{Two targets}
Consider initially that only targets one and two are present. As we observed earlier, in this case the detection probability depends only on the false alarm probability, the receive \gls{snr}, and the coherence, $C_{1,2}$.
In \cref{fig:coherence_plots}, we show the coherence of the \gls{ris} and non-\gls{ris} signals, as a function of $\Delta$ with varying number of \gls{ris} elements, \gls{ris} phase profiles, and number of \gls{ue} antennas: we see that the size of the \gls{ris} is crucial for fast decoherence while the importance of the number of \gls{ris} phase profiles is less pronounced. Moreover, we generally observe a lower coherence with the \gls{ris} signal than with the non-\gls{ris} signal for all tested specifications, revealing the feasible detection performance gains by using the \gls{ris}. The connection to detection probability was detailed in \cref{fig:theoretical_coherence_plots}.

\subsubsection{Three targets}
Consider now the \gls{sp} cluster of three targets. The detection probability in this case not only depends on the coherences between the targets, but also $\Tilde{C}$, cf.~\cref{def:generalized_coherence}. We plot in \cref{fig:TildeC} $\Tilde{C}$ as a function of $\Delta$: as was the case for $C_{1,2}$ in \cref{fig:coherence_plots}, the \gls{ris} signal with $N_{\rm r} = 35 \times 35$ shows significant improvements in decoherence as compared to the non-\gls{ris} signal, and this is particularly pronounced when $\Tilde{T} = 25$.

Using the results for the expected detection probability \cref{eq:expected_detection_parameter1_L3,eq:expected_detection_parameter2_L3,eq:expected_detection_parameter3_L3}, including \cref{prop:expected_joint_detection_probability}, we show in \cref{fig:Epmd_ROC} the expected misdetection probability against the false alarm probability on a log-log scale for each of the three targets when $\Delta = 0.1$ rad. This shows that for these settings, the \gls{ris} signal has improved detection capabilities compared to the non-\gls{ris} signal for both the second and third targets, although the non-\gls{ris} signal has better detection capabilities for the first target, due to the higher \gls{snr}.
Moreover, we see that joint detection outperforms individual detection using either \gls{ris} signal or non-\gls{ris} signal for all three targets.

Comparing the detection performance depending on the spacing between the targets, $\Delta$, we show in \cref{fig:Epd_spacing} the expected \gls{auc} for varying specifications regarding number of \gls{ris} elements, $N_{\rm r}$, and the number of \gls{ris} phase profiles, $\Tilde{T}$, as well as the \gls{ue} specifications, herein the number of \gls{ue} antennas.
It is observed that we can achieve improved detection performance with the \gls{ris} signal over the non-\gls{ris} signal, in cases where $N_{\rm u}$ is small while $N_{\rm r}$ is relatively large: with the \gls{ris} signal, when $\Tilde{T}=25$ and $N_{\rm r} = 35\times 35$, an \gls{auc} of $1$ is approached at $\Delta=0.02~{\rm rad}$, while for the non-\gls{ris} signal, with $N_{\rm u}=2\times 2$, this occurs at $\Delta=0.1~{\rm rad}$.

%
\begin{figure*}[t]
    \centering
    \begin{minipage}{0.3\textwidth}
        \centering
        \subfloat[AUC target 2.]{\includegraphics[width=\linewidth, page=16]{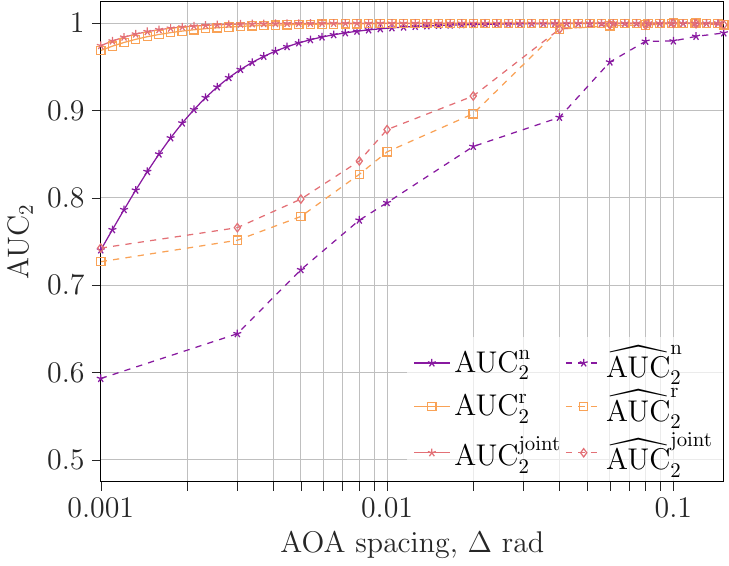}\label{subfig:Sensing_AUC2}}
    \end{minipage}
    ~
    \begin{minipage}{0.3\textwidth}
        \centering
        \subfloat[AUC target 3.]{\includegraphics[width=\linewidth, page=17]{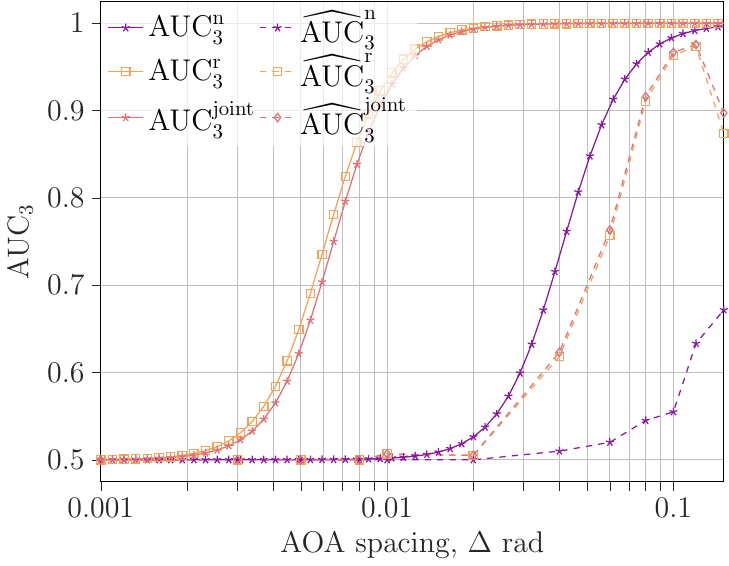}\label{subfig:Sensing_AUC3}}
    \end{minipage}
    ~
    \begin{minipage}{0.3\textwidth}
        \centering
        \subfloat[OSPA.]{\includegraphics[width=\linewidth, page=18]{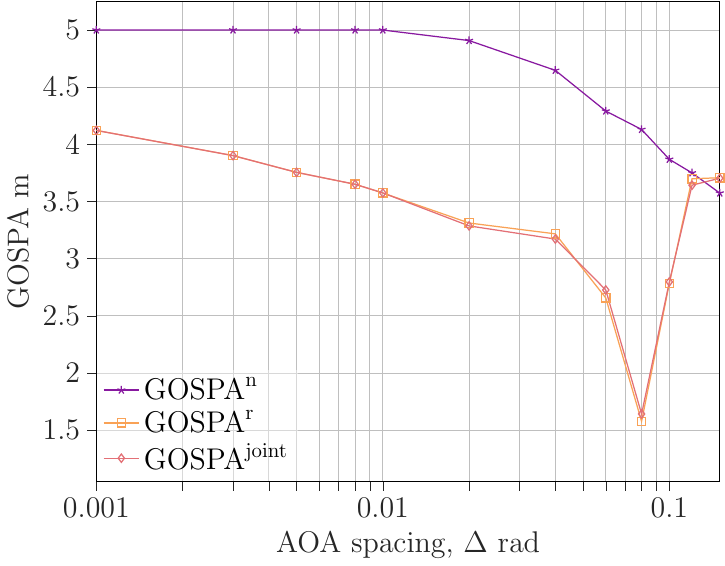}\label{subfig:Sensing_OSPA}}
    \end{minipage}
    \caption{Results using the channel estimation algorithm for varying \gls{aoa} spacing: (a,b) comparison of empirical expected \gls{auc} and the theoretical upper bound on the \gls{auc}, and (c) the \gls{gospa} metric.}
    \label{fig:Sensing}
\end{figure*}

\subsection{Positioning Bound}
\noindent With the same setup as in the previous section, we show in \cref{fig:FisherCDF} the empirical \glspl{cdf} of the \gls{peb}, \gls{deb}, and \gls{aeb}, derived in \cref{subsec:fisher_analysis}, for each of the three targets when $\Delta=0.1~{\rm rad}$ and considering the non-\gls{ris}, \gls{ris}, and combined signals. From \cref{subfig:PEBcdf} we see that combining the non-\gls{ris} and \gls{ris} signals in all cases outperforms the individual ones in terms of \gls{peb}, and we also see that the \gls{ris} signal performs better than the non-\gls{ris} signal. Going to the \glspl{deb} and \glspl{aeb} in \cref{subfig:DEBcdf,subfig:AEBcdf}, we notice that the non-\gls{ris} delay is better estimated than the \gls{ris} delay which is due to better \gls{snr} conditions for the non-\gls{ris} signal, but on the other hand the angle is better estimated by the \gls{ris} than the non-\gls{ris} signal due to the higher angular resolution provided by the \gls{ris}.

\subsection{Sensing Performance}
\noindent In \cref{subfig:Sensing_AUC2,subfig:Sensing_AUC3}, we show the expected \gls{auc} for the sensing algorithm with the \gls{omp} channel estimation method in comparison to the derived bound as a function of the \gls{aoa} spacing, $\Delta$ rad. We observe that while we do not achieve the bound, the relation between the \gls{ris}, non-\gls{ris}, and joint methods follow a similar trend as with the bound. Moreover, as in the bound, we observe that the \gls{auc} with the \gls{ris} signal dominates that of the non-\gls{ris} signal, while the combination of the two signals has similar performance to only using the \gls{ris} signal. Focusing on $\Delta=0.1$ rad, we are able to improve the \gls{auc} for the third target from $0.55$ without the \gls{ris} to $0.95$ with the \gls{ris}.

Based on preliminary calibration, we fix the threshold parameters
${\rm residual}_{\rm th}^{\rm n} = 3\cdot 10^{-5}$, and ${\rm residual}_{\rm th}^{\rm r} = 3.5\cdot 10^{-5}$, and show the \gls{gospa} for varying \gls{aoa} spacing, $\Delta$ rad, in \cref{subfig:Sensing_OSPA}. Here, we observe that the \gls{ris} signal has a lower \gls{gospa} than the non-\gls{ris} signal for small $\Delta$ with largest improvement of $2.5$ m at $\Delta=0.8$, however, when $\Delta > 0.1$ the \gls{ris} performance degrades. This happens due to two reasons: (i) as $\Delta$ increases the resolution region construction can fail to cover all the \glspl{sp}, and (ii) as $\Delta$ increases the resolution region tends to grow causing less concentration of the \gls{ris} beams resulting in a lower receive \gls{snr}.

\section{Conclusion}\label{sec:conclusion}
\noindent In this work, we have assessed the potential of \gls{ris} for sensing targets in close proximity in the presence of a \gls{los} link. With a theoretical analysis of the sensing system, considering detection probability and Fisher analysis, we have uncovered the importance of a novel coherence concept. We can conclude that as \gls{ris} can improve the sensing resolution if the \gls{ris}-assisted signal decoheres more rapidly than the non-\gls{ris} signal, while maintaining a sufficiently high receive \gls{snr}. The results show the capability to achieve a lower coherence with the \gls{ris}-assisted signal than the non-\gls{ris} signal when the sensor has a small antenna array relative to the size of the \gls{ris} and the number of different \gls{ris} phase profiles. This improvement in terms of coherence explains the observed improvements in sensing performance in terms of both detection and localization accuracy with the \gls{ris}-assisted signal as opposed to the non-\gls{ris} signal. Finally, we find that combining the two signals can lead to sensing improvements if the signals are of similar quality. In future work we will use the coherence concept to guide optimization of the \gls{ris} phase profiles, which may further improve high-resolution sensing capabilities of the \gls{ris}-assisted signal.

\appendices

\section{Channel Parameters}\label{subsec:channel_parameters}
\noindent The complex path gains $\alpha_0$, $\bar{\alpha}_l$, and $\alpha_l$ are modelled as $\bar{\alpha}_l = P_{\bar{\alpha}_l} \bar{\xi}_l$ and $\alpha_l = P_{\alpha_l} \xi_l$ where $\bar{\xi}_l,\xi_l \sim \mathcal{C}\mathcal{N}(0, \frac{4\sigma^2_{\rm rcs,l}}{\pi})$ are independent and simulates a random phase shift and a Rayleigh distributed square root of \gls{rcs} with mean $\sigma_{\rm rcs}$. This also means that $|\alpha_l|$ (resp. $\bar{\alpha}_l$) is Rayleigh distributed with scale parameter $\varsigma_{\alpha_l} := \sqrt{\frac{2}{\pi}}P_{\alpha_l} \sigma_{\rm rcs}$ (resp. $\varsigma_{\bar{\alpha}_l} := \sqrt{\frac{2}{\pi}}P_{\bar{\alpha}_l} \sigma_{\rm rcs}$). The amplitudes are modelled as 
\begin{subequations}\label{eq:ch_gains}
\begin{align}
    P_{\alpha_0}^2 &= \frac{E_s\lambda^4 (g^{\rm ur})^{4q_0}}{(4\pi)^3 \Vert\bs{p} - \bs{p}_{\rm r}\Vert^4},\label{eq:single_alpha0}\\
    P_{\alpha_l}^2 &= \frac{E_s\lambda^2}{(4\pi)^3 \Vert\bs{c}_l - \bs{p}\Vert^4}, \label{eq:single_beta1}\\
    P_{\bar{\alpha}_l}^2 &= \frac{E_s\lambda^4 (g^{\rm ur})^{2q_0}(g_l^{\rm sr})^{2q_0}}{(4\pi)^4 \Vert\bs{p} - \bs{p}_{\rm r}\Vert^2 \Vert\bs{p}_{\rm r} - \bs{c}_l\Vert^2 \Vert\bs{c}_l - \bs{p}\Vert^2},\label{eq:single_alpha1}
\end{align}
\end{subequations}
where $E_{\rm s}$ is the energy per subcarrier, $q_0=0.285$, $g^{\rm ur}=\frac{(\bs{p} - \bs{p}_{\rm r})^\top\bs{n}_{\rm r}}{\Vert\bs{p} - \bs{p}_{\rm r}\Vert}$, $g_l^{\rm sr}=\frac{(\bs{c}_l - \bs{p}_{\rm r})^\top\bs{n}_{\rm r}}{\Vert\bs{c}_l - \bs{p}_{\rm r}\Vert}$, and $\bs{n}_{\rm r}$ is the normal vector of the \gls{ris} \cite{Hyowon2023:Bounce,AbuShaban2018Bounds}.
The \glspl{toa} $\tau_0$, $\bar{\tau}_l$, and $\tau_l$ and the \glspl{aoa} $\bs{\phi}_0$, $\bs{\phi}_l$, $\bs{\theta}_0$, and $\bs{\theta}_l$ are defined as
\begin{subequations}\label{eq:ch_params}
\begin{align}
    \tau_0 &= \frac{2\Vert \bs{p} - \bs{p}_{\rm r}\Vert}{c},\label{eq:RIS_UE_delay}\\
    \bs{\theta}_0 &= \big(\text{atan2}(y^{\rm ru}, x^{\rm ru}), \text{arccos}(z^{\rm ru})\big)^\top,\label{eq:RIS_UE_angle}\\
    \bs{\phi}_0 &= \big(\text{atan2}(y^{\rm ur}, x^{\rm ur}), \text{arccos}(z^{\rm ur})\big)^\top,\label{eq:UE_RIS_angle}\\
    \tau_l &= \tau(\bs{c}_l) = \frac{2 \Vert \bs{c}_l - \bs{p}\Vert}{c},\label{eq:nonRIS_delay}\\
    \bar{\tau}_l &= \bar{\tau}(\bs{c}_l) = \frac{\Vert \bs{p} - \bs{p}_{\rm r}\Vert + \Vert \bs{p}_{\rm r} - \bs{c}_l\Vert + \Vert \bs{c}_l - \bs{p}\Vert}{c},\label{eq:RIS_delay}\\
    \bs{\theta}_l &= \bs{\theta}(\bs{c}_l) = \big(\text{atan2}(y_l^{\rm su}, x_l^{\rm su}), \text{arccos}(z_l^{\rm su})\big)^\top,\label{eq:nonRIS_angle}\\
    \bs{\phi}_l &= \bs{\phi}(\bs{c}_l) = \big(\text{atan2}(y_l^{\rm sr}, x_l^{\rm sr}), \text{arccos}(z_l^{\rm sr})\big)^\top,\label{eq:RIS_angle}
\end{align}
\end{subequations}
where $\bs{x}^{\rm ru} = (x^{\rm ru}, y^{\rm ru}, z^{\rm ru})^\top$, $\bs{x}_l^{\rm su} = (x_l^{\rm su}, y_l^{\rm su}, z_l^{\rm su})^\top$, $\bs{x}^{\rm ur} = (x^{\rm ur}, y^{\rm ur}, z^{\rm ur})^\top$, and $\bs{x}_l^{\rm sr} = (x_l^{\rm sr}, y_l^{\rm sr}, z_l^{\rm sr})^\top$ are the direction vectors of the \gls{ris} to the \gls{ue}, the $l$th \gls{sp} to the \gls{ue}, the \gls{ue} to the \gls{ris}, and the $l$th \gls{sp} to the \gls{ris} in the respective local coordinate system. These vectors can be expressed using global positions $\bs{c}_l$, $\bs{p}$, $\bs{p}_{\rm r}$ and rotation matrices $\bs{Q}_{\rm u}$ and $\bs{Q}_{\rm r}$ as $\bs{x}^{\rm ru} = \bs{Q}_{\rm u}^\top \frac{\bs{p}_{\rm r} - \bs{p}}{\Vert\bs{p}_{\rm r} - \bs{p}\Vert}$, $\bs{x}_l^{\rm su} = \bs{Q}_{\rm u}^\top \frac{\bs{c}_l - \bs{p}}{\Vert\bs{c}_l - \bs{p}\Vert}$, $\bs{x}_l^{\rm sr} = \bs{Q}_{\rm r}^\top \frac{\bs{c}_l - \bs{p}_{\rm r}}{\Vert\bs{c}_l - \bs{p}_{\rm r}\Vert}$, and $\bs{x}^{\rm ur} = \bs{Q}_{\rm r}^\top \frac{\bs{p} - \bs{p}_{\rm r}}{\Vert\bs{p} - \bs{p}_{\rm r}\Vert}$ \cite{AbuShaban2018Bounds}.
The function $\text{atan2}(x, y)$ returns the principal value of the argument of $x+jy$.

The expressions \cref{eq:RIS_delay,eq:nonRIS_delay,eq:nonRIS_angle,eq:RIS_angle} define a mapping from the \gls{sp} position $\bs{c}_l$ into the channel parameter space $(\tau_l, \bar{\tau}_l, \bs{\theta}_l, \bs{\phi}_l)$, specifically, we define $\bs{h}(\bs{c}) = [\tau(\bs{c}), \bs{\theta}^\top(\bs{c}), \bar{\tau}(\bs{c}), \bs{\phi}^\top(\bs{c})]^\top$.

\section{Proof of results in Section III.}
\subsection{Proof of proposition 1}\label{subsec:proof_prop1}
\begin{proof}\label{proof:proof_prop1}
    Conditioned on the channel coefficients we have that $\bs{y}|\bs{\alpha} \sim \mathcal{CN}(\bs{G} \bs{\alpha}, \frac{\sigma^2}{2}\bs{I})$.
    In the $l$-th iteration for $l=2,\dots,L$, given that we know $\bs{\eta}_1,\dots, \bs{\eta}_{l-1}$, the maximum likelihood estimate of $\bs{\alpha}_{1:l-1}$ is
    \begin{align}
        \hat{\bs{\alpha}}_{l,1:l-1} &= \big(\bs{G}_{1:l-1}^H\bs{G}_{1:l-1}\big)^{-1} \bs{G}_{1:l-1}^H \bs{y} \nonumber\\
        &= \bs{\alpha}_{1:l-1} + (\bs{G}_{1:l-1}^{\rm n})^\dagger \big(\bs{G}_{l:L}^{\rm n}\bs{\alpha}_{l:L} + \bs{\varepsilon}^{\rm n}\big)\label{eq:OMP_alpha_est}
    \end{align}
    where $\bs{G}_{1:l-1}^\dagger := \big(\bs{G}_{1:l-1}^H\bs{G}_{1:l-1}^{\rm n}\big)^{-1} \bs{G}_{1:l-1}^H$ is the Moore-Penrose pseudo-inverse of $\bs{G}_{1:l-1}$. 
    Following the greedy approach of \gls{omp}, we define a residual as
    \begin{equation}
        \bs{r}_l := \bs{y} - \bs{G}_{1:l-1} \hat{\bs{\alpha}}_{l,1:l-1} = \big(\bs{I} - \bs{P}_{l-1}\big)(\bs{G}_{l:L} \bs{\alpha}_{l:L} + \bs{\varepsilon})
    \end{equation}
    where $\bs{P}_{l-1}$ is the projection matrix for the maximum likelihood fitting, i.e., $\bs{P}_{l-1} = \bs{G}_{1:l-1} \bs{G}_{1:l-1}^\dagger$.
    Accordingly, $(\bs{I} - \bs{P}_{l-1})$ is itself a projection matrix onto the orthogonal space. Hence, conditioned on the channel coefficient we have that
    \begin{equation}
            \bs{r}_l|\bs{\alpha}_{l:L} \sim \mathcal{CN}\Big(\big(\bs{I} - \bs{P}_{l-1}\big) \bs{G}_{l:L} \bs{\alpha}_{l:L}, \frac{\sigma^2}{2}\big(\bs{I} - \bs{P}_{l-1}\big)\Big).
    \end{equation}
    Assume now that $\bs{\eta}_l$ is known. The least squares estimate of the channel coefficient is then
    $\hat{\alpha}_l = \frac{\bs{g}_l^H \bs{r}_l}{\Vert \bs{g}_l\Vert^2}$, hence $\hat{\alpha}_l|\bs{\alpha}_{l:L}$ is circularly-symmetric complex Gaussian with $\E[\hat{\alpha}_l] = \frac{\bs{g}_l^H\big(\bs{I} - \bs{P}_{l-1}\big) \bs{G}_{l:L}\bs{\alpha}_{l:L}}{\Vert \bs{g}_l\Vert^2}$ and ${\rm Var}[\hat{\alpha}_l] = \frac{\sigma^2 \bs{g}_l^H\big(\bs{I} - \bs{P}_{l-1}\big) \bs{g}_l}{2 \Vert \bs{g}_l\Vert^4}$. Hence, the statistic $\gamma_l := \frac{2|\hat{\alpha}_l|^2}{{\rm Var}[\hat{\alpha}_l]}$ follows a non-central $\chi^2$ distribution with $2$ degrees of freedom and non-centrality parameter
    \begin{equation}
        \mu_l := \frac{2|\E[\hat{\alpha}_l]|^2}{{\rm Var}[\hat{\alpha}_l]} = \frac{4|\bs{g}_l^H\big(\bs{I} - \bs{P}_{l-1}\big) \bs{G}_{l:L}\bs{\alpha}_{l:L}|^2}{\sigma^2\bs{g}_l^H\big(\bs{I} - \bs{P}_{l-1}\big) \bs{g}_l}
    \end{equation}
    and the result follows immediately by observing that $p_{{\rm d}, l}(\bs{\alpha}_{l:L}) = {\rm Pr}(\gamma_l \geq \gamma_{\rm th} | \bs{\alpha}_{l:L}) = Q_1(\sqrt{\mu_l}, \sqrt{\gamma_{\rm th}})$. 
\end{proof}
\subsection{Proof of corollary 1}\label{subsec:proof_corr1}
\begin{proof}\label{proof:proof_corr1}
    Let $\beta_l := |\bs{g}_l^H\big(\bs{I} - \bs{P}_{l-1}\big) \bs{G}_{l:L}\bs{\alpha}_{l:L}|^2$ and notice that the distributional assumptions of $\alpha_l$ implies that $\sqrt{\beta_l}$ is Rayleigh distributed with squared scale parameter $\zeta_l$.
    Define $A_l := \frac{4}{\sigma^2\bs{g}_l^H\big(\bs{I} - \bs{P}_{l-1}\big) \bs{g}_l}$ such that $\mu_l = A_l \beta_l$. The marginal distribution for detection, also interpreted as the expected detection probability is $p_{{\rm d}, l} = {\rm Pr}(\gamma_l \geq \gamma_{\rm th})$.
    By the law of total probability, $p_{{\rm d}, l}$ can be computed as
    \begin{align}
        p_{{\rm d}, l} &= \int_0^\infty Q_1(\sqrt{A_l \beta_l}, \sqrt{\gamma_{\rm th}}) \frac{\beta_l}{\zeta_l} {\rm exp}\Big(\frac{-\beta_l}{2\zeta_l}\Big) {\rm d}\beta_l\nonumber\\
        &= \int_0^\infty x Q_1(\sqrt{A_l \zeta_l} x, \sqrt{\gamma_{\rm th}}) {\rm exp}\Big(\frac{-x^2}{2}\Big) {\rm d}x,
    \end{align}
    using variable substitution $x := \sqrt{\frac{\beta_l}{\zeta_l}}$. Then
    \begin{equation}
        p_{{\rm d}, l} = {\rm exp}\Big(\frac{-\gamma_{\rm th}}{2 A_l \zeta_l + 2}\Big) = {\rm exp}\Big(\frac{{\rm log}(p_{\rm fa})}{A_l \zeta_l + 1}\Big),
    \end{equation}
    is found by computing the integral \cite{Sofotasios2014:Analytic}.
\end{proof}
\subsection{Proof of proposition 2}\label{subsec:proof_prop2}
\begin{proof}\label{proof:proof_prop2}
    Due to the assumption of independence of the channel coefficients, the joint distribution of $\beta_l$ and $\bar{\beta}_l$ factorizes, and the expected joint detection probability is
    \begin{align}
        p_{{\rm d}, l}^{\rm joint} &= \E_{\bar{\beta}_l}[\E_{\beta_l}[p_{{\rm d}, l}^{\rm joint}(\bs{\alpha}_{l:L}, \bar{\bs{\alpha}}_{l:L})]]\nonumber\\
        &= \E_{\bar{\beta}_l}[p_{{\rm d}, l}^{\rm n} + p_{{\rm d}, l}^{\rm r}(\bar{\bs{\alpha}}_{l:L}) - p_{{\rm d}, l}^{\rm n}p_{{\rm d}, l}^{\rm r}(\bar{\bs{\alpha}}_{l:L})]\nonumber\\
        &= p_{{\rm d}, l}^{\rm n} + p_{{\rm d}, l}^{\rm r} - p_{{\rm d}, l}^{\rm n} p_{{\rm d}, l}^{\rm r},
    \end{align}
    where $\E_{\beta_l}$ and $\E_{\bar{\beta}_l}$ denotes expectation with respect to $\beta_l$ and $\bar{\beta}_l$, respectively.
    The joint false alarm probability follows immediately by the same arguments.
\end{proof}

\section{Fisher Information Derivations}\label{subsec:fisher_information_derivations}
\noindent In this appendix, we provide the partial derivatives needed in the Fisher analysis \cref{subsec:fisher_analysis}.
Before providing the partial derivatives, for convenience of notation, we begin by defining
\begin{subequations}
\begin{align}
    \bs{g}_{\tau_l}^{\rm n} &= \langle \bs{a}_{\rm u}^*(\bs{\theta}_l), f \rangle \bs{1}_{\Tilde{T}} \otimes \big(\bs{n} \odot \bs{d}(\tau_l)\big) \otimes \bs{a}_{\rm u}(\bs{\theta}_l),\\
    \begin{split}
        \bs{g}_{\theta_l^{\rm az}}^{\rm n} &=  \bs{1}_{\Tilde{T}} \otimes \bs{d}(\tau_l) \otimes \Big(\langle \bs{a}_{\rm u}^*(\bs{\theta}_l), f \rangle \frac{\partial \bs{a}_{\rm u}(\bs{\theta}_l)}{\partial \theta_l^{\rm az}}\\ &\qquad\qquad\qquad\qquad + \langle \frac{\partial \bs{a}_{\rm u}^*(\bs{\theta}_l)}{\partial \theta_l^{\rm az}}, f\rangle \bs{a}_{\rm u}(\bs{\theta}_l)\Big).
    \end{split}\\
    \begin{split}
        \bs{g}_{\theta_l^{\rm el}}^{\rm n} &= \bs{1}_{\Tilde{T}} \otimes \bs{d}(\tau_l) \otimes \Big(\langle \bs{a}_{\rm u}^*(\bs{\theta}_l), f \rangle \frac{\partial \bs{a}_{\rm u}(\bs{\theta}_l)}{\partial \theta_l^{\rm el}}\\ &\qquad\qquad\qquad\qquad + \langle \frac{\partial \bs{a}_{\rm u}^*(\bs{\theta}_l)}{\partial \theta_l^{\rm el}}, f\rangle \bs{a}_{\rm u}(\bs{\theta}_l)\Big).
    \end{split}
\end{align}
\end{subequations}
where we have defined $\bs{n} = [0, 1, \dots, N-1]^\top$, and where the partial derivatives of the array response vector are given as
\begin{subequations}
\begin{align}
    \frac{\partial \bs{a}_{\rm u}(\bs{\theta}_l)}{\partial \theta_l^{\rm az}} &= j \Big(\bs{P}_{\rm u}^\top \frac{\partial \bs{\kappa}(\bs{\theta}_l)}{\partial \theta_l^{\rm az}}\Big) \odot \exp(j \bs{P}_{\rm u}^\top \bs{\kappa}(\bs{\theta}_l)),\label{eq:derArrayUaz}\\
    \frac{\partial \bs{a}_{\rm u}(\bs{\theta}_l)}{\partial \theta_l^{\rm el}} &= j \Big(\bs{P}_{\rm u}^\top \frac{\partial \bs{\kappa}(\bs{\theta}_l)}{\partial \theta_l^{\rm el}}\Big) \odot \exp(j \bs{P}_{\rm u}^\top \bs{\kappa}(\bs{\theta}_l)),\label{eq:derArrayUel}
\end{align}
\end{subequations}
\resizebox{1\linewidth}{!}{with partial derivatives of the wavenumber vector given as~~}
$\frac{\partial \bs{\kappa}(\bs{\theta}_l)}{\partial \theta_l^{\rm az}} = \frac{2\pi}{\lambda} [-\sin(\theta_l^{\rm az})\sin(\theta_l^{\rm el}), \cos(\theta_l^{\rm az})\sin(\theta_l^{\rm el}), 0]^\top$, and \resizebox{1\linewidth}{!}{$\frac{\partial \bs{\kappa}(\bs{\theta}_l)}{\partial \theta_l^{\rm el}} = \frac{2\pi}{\lambda} [\cos(\theta_l^{\rm az})\cos(\theta_l^{\rm el}), \sin(\theta_l^{\rm az})\cos(\theta_l^{\rm el}), -\sin(\theta_l^{\rm el})]^\top$.}
Now the partial derivatives for the non-\gls{ris} signal are simply given by $\frac{\partial \bs{\mu}^{\rm n}}{\partial \Re\{\alpha_l\}} = \bs{g}_{l}^{\rm n}$, $\frac{\partial \bs{\mu}^{\rm n}}{\partial \Im\{\alpha_l\}} = j \bs{g}^{\rm n}(\tau_l, \bs{\theta}_l)$, $\frac{\partial \bs{\mu}^{\rm n}}{\partial \tau_l} = -j2\pi \Delta_f \alpha_l \bs{g}_{\tau_l}^{\rm n}$, $\frac{\partial \bs{\mu}^{\rm n}}{\partial \theta_l^{\rm az}} = \alpha_l \bs{g}_{\theta_l^{\rm az}}^{\rm n}$, and $\frac{\partial \bs{\mu}^{\rm n}}{\partial \theta_l^{\rm el}} = \alpha_l\bs{g}_{\theta_l^{\rm el}}^{\rm n}$.

We define for convenience the notation
\begin{subequations}
\begin{align}
    \bs{g}_{\tau_l}^{\rm r} &= \langle \bs{a}_{\rm u}^*(\bs{\theta}_l), f \rangle \bs{\nu}(\bs{\phi}_l) \otimes \big(\bs{n} \odot \bs{d}(\bar{\tau}_l)\big) \otimes \bs{a}_{\rm u}(\bs{\theta}_0),\\
    \bs{g}_{\theta_l^{\rm az}}^{\rm r} &= \langle \frac{\partial \bs{a}^*_{\rm u}(\bs{\theta}_l)}{\partial \theta_l^{\rm az}}, \bs{f}\rangle \bs{\nu}(\bs{\phi}_l) \otimes \bs{d}(\bar{\tau}_l) \otimes \bs{a}_{\rm u}(\bs{\theta}_0),\\
    \bs{g}_{\theta_l^{\rm el}}^{\rm r} &= \langle \frac{\partial \bs{a}^*_{\rm u}(\bs{\theta}_l)}{\partial \theta_l^{\rm el}}, \bs{f}\rangle \bs{\nu}(\bs{\phi}_l) \otimes \bs{d}(\bar{\tau}_l) \otimes \bs{a}_{\rm u}(\bs{\theta}_0),\\
    \begin{split}
        \bs{g}_{\phi_l^{\rm az}}^{\rm r} &= \langle \bs{a}_{\rm u}^*(\bs{\theta}_l), f \rangle \Big(\big(\bs{\omega}_t \odot \bs{a}_{\rm r}(\bs{\phi}_0)\big)^\top \frac{\partial \bs{a}_{\rm r}(\bs{\phi}_l)}{\partial \phi_l^{\rm az}}\Big)\\ &\qquad\qquad\qquad\qquad \otimes \bs{d}(\bar{\tau}_l) \otimes \bs{a}_{\rm u}(\bs{\theta}_0).
    \end{split}\\
    \begin{split}
        \bs{g}_{\phi_l^{\rm el}}^{\rm r} &= \langle \bs{a}_{\rm u}^*(\bs{\theta}_l), f \rangle \Big(\big(\bs{\omega}_t \odot \bs{a}_{\rm r}(\bs{\phi}_0)\big)^\top \frac{\partial \bs{a}_{\rm r}(\bs{\phi}_l)}{\partial \phi_l^{\rm el}}\Big)\\ &\qquad\qquad\qquad\qquad \otimes \bs{d}(\bar{\tau}_l) \otimes \bs{a}_{\rm u}(\bs{\theta}_0).
    \end{split}
\end{align}
\end{subequations}
where the partial derivatives of the array response vector at the \gls{ris} is computed as in \cref{eq:derArrayUaz,eq:derArrayUel}.
Now, the partial derivatives for the \gls{ris} signal are easily represented as $\frac{\partial \bs{\mu}^{\rm r}}{\partial \Re\{\bar{\alpha}_l\}} = \bs{g}_{l}^{\rm r}$, $\frac{\partial \bs{\mu}^{\rm r}}{\partial \Im\{\bar{\alpha}_l\}} = j\bs{g}_{l}^{\rm r}$, $\frac{\partial \bs{\mu}^{\rm r}}{\partial \bar{\tau}_l} = -j2\pi \Delta_f \bar{\alpha}_l \bs{g}_{\tau_l}^{\rm r}$, $\frac{\partial \bs{\mu}^{\rm r}}{\partial \theta_l^{\rm az}} = \bar{\alpha}_l \bs{g}_{\theta_l^{\rm az}}^{\rm r}$, $\frac{\partial \bs{\mu}^{\rm r}}{\partial \theta_l^{\rm el}} = \bar{\alpha}_l \bs{g}_{\theta_l^{\rm el}}^{\rm r}$, $\frac{\partial \bs{\mu}^{\rm r}}{\partial \phi_l^{\rm az}} = \bar{\alpha}_l \bs{g}_{\phi_l^{\rm az}}^{\rm r}$, and $\frac{\partial \bs{\mu}^{\rm r}}{\partial \phi_l^{\rm el}} = \bar{\alpha}_l \bs{g}_{\phi_l^{\rm el}}^{\rm r}$.

We provide the partial derivatives used to compute the Jacobian matrices, $\bs{T}^{\rm n}$ and $\bs{T}^{\rm r}$:
\begin{subequations}\label{eq:Fisher_Jacobian}
\begin{align}
    \frac{\partial \tau(\bs{c}_l)}{\partial \bs{c}_l} &= \frac{2(\bs{c}_l - \bs{p})}{c\Vert \bs{c}_l - \bs{p}\Vert},\\
    \frac{\partial \bar{\tau}(\bs{c}_l)}{\partial \bs{c}_l} &= \frac{\bs{p}_{\rm r} - \bs{c}_l}{c \Vert \bs{p}_{\rm r} - \bs{c}_l\Vert} + \frac{\bs{c}_l - \bs{p}}{c \Vert \bs{p} - \bs{c}_l\Vert},\\
    \frac{\partial \theta^{\rm az}(\bs{c}_l)}{\partial \bs{c}_l} &= \frac{x_l^{\rm su}[\bs{Q}_{\rm u}]_{:,2} - y_l^{\rm su}[\bs{Q}_{\rm u}]_{:,1}}{(x_l^{\rm su})^2 + (y_l^{\rm su})^2},\\
    \frac{\partial \theta^{\rm el}(\bs{c}_l)}{\partial \bs{c}_l} &= -\frac{\Vert \bs{x}_l^{\rm su}\Vert^2 [\bs{Q}_{\rm u}]_{:,3} - z_l^{\rm su}(\bs{c}_l - \bs{p})}{\Vert \bs{x}_l^{\rm su}\Vert^3 \Big(1 - \Big(\frac{z_l^{\rm su}}{\Vert \bs{x}_l^{\rm su}\Vert}\Big)^2\Big)^{1/2}},\\
    \frac{\partial \phi^{\rm az}(\bs{c}_l)}{\partial \bs{c}_l} &= \frac{x_l^{\rm sr} [\bs{Q}_{\rm r}]_{:,2} - y_l^{\rm sr} [\bs{Q}_{\rm r}]_{:,1}}{(x_l^{\rm sr})^2 + (y_l^{\rm sr})^2},\\
    \frac{\partial \phi^{\rm el}(\bs{c}_l)}{\partial \bs{c}_l} &= -\frac{\Vert \bs{x}_l^{\rm sr}\Vert^2 [\bs{Q}_{\rm r}]_{:,3} - z_l^{\rm sr}(\bs{c}_l - \bs{p}_{\rm r})}{\Vert \bs{x}_l^{\rm sr}\Vert^3\Big(1 - \Big(\frac{z_l^{\rm sr}}{\Vert \bs{x}_l^{\rm sr}\Vert}\Big)^2\Big)^{1/2}}.
\end{align}
\end{subequations}

\printbibliography

\end{document}